\documentclass[11pt]{article}
\usepackage[margin=1in]{geometry}

\usepackage{booktabs} % For formal tables
\usepackage[ruled]{algorithm2e} % For algorithms

\usepackage{amsmath,amssymb,amsthm}
\usepackage{graphicx} % Required for inserting images
\usepackage{xcolor}
\usepackage{tikz}
\usepackage{calc}
\usepackage{pgffor}
\usetikzlibrary{calc}
\usetikzlibrary{math}
\usepackage{xspace}
\usepackage{hyperref}
\usepackage{cite}

\newtheorem{observation}{Observation}
\newtheorem{lemma}{Lemma}
\newtheorem{theorem}{Theorem}

\newtheorem{definition}{Definition}
\newtheorem{example}{Example}

\newcommand{\calA}{\ensuremath{\mathcal{A}}\xspace}
\newcommand{\calP}{\ensuremath{\mathcal{P}}\xspace}

\newcommand{\calO}{\ensuremath{\mathcal{O}}\xspace}

\newcommand{\faircap}[1]{\smash{\varphi_{#1}}\xspace}
\newcommand{\faircapP}[1]{\smash{\varphi^{\calP}_{#1}}\xspace}
\newcommand{\faircapA}[1]{\smash{\varphi^{\calA}_{#1}}\xspace}

\allowdisplaybreaks

\renewcommand{\paragraph}[1]{{\medskip \noindent \bfseries \textsf{#1} \hspace{0.05cm}}}

\title{Designing Exploration Contracts}

\author{Martin Hoefer\thanks{Department of Computer Science, RWTH Aachen University, Germany, {\tt mhoefer@cs.rwth-aachen.de}. Supported by DFG Research Unit ADYN (project number 411362735) and DFG grant Ho 3831/9-1 (project number 514505843).} \and Conrad Schecker\thanks{Institute for Computer Science, Goethe University Frankfurt, Germany, {\tt schecker@em.uni-frankfurt.de}.}  \and Kevin Schewior\thanks{University of Southern Denmark, Denmark, {\tt kevs@sdu.dk}. Supported by the Independent Research Fund Denmark, Natural Sciences, grant DFF-0135-00018B.}}

\begin{document}
\maketitle

\begin{abstract}
We study a natural application of contract design in the context of sequential exploration problems.
In our principal--agent setting, a search task is delegated to an agent. The agent performs a sequential exploration of $n$ boxes, suffers the exploration cost for each inspected box, and selects the content (called the \emph{prize}) of one inspected box as outcome. Agent and principal obtain an individual value based on the selected prize. To influence the search, the principal a-priori designs a contract with a non-negative payment to the agent for each potential prize. The goal of the principal is to maximize her expected reward, i.e., value minus payment. Interestingly, this natural contract scenario shares close relations with the \emph{Pandora's Box} problem. 

We show how to compute optimal contracts for the principal in several scenarios. A popular and important subclass is that of \emph{linear} contracts, and we show how to compute optimal linear contracts in polynomial time. For general contracts, we obtain optimal contracts under the standard assumption that the agent suffers cost but obtains value only from the transfers by the principal. More generally, for general contracts with non-zero agent values for outcomes we show how to compute an optimal contract in two cases: (1) when each box has only one prize with non-zero value for principal and agent, (2) for i.i.d.\ boxes with a single prize with positive value for the principal.
\end{abstract}

\section{Introduction}
In many real-world situations, e.g., the search for a job candidate or house, a decision maker is faced with the choice between different alternatives of a-priori unknown value, which can be explored at some cost. Due to missing qualifications or time constraints, in many markets such exploration tasks are not executed directly by the decision maker. Rather, a decision maker (called \emph{principal} \calP in the following) can delegate the exploration to an \emph{agent} \calA{}, whose incentives are potentially misaligned with that of \calP.

Indeed, suppose the principal intends to buy a house. They might know an inspection cost and have some prior knowledge about the value of each house. However, they might not be qualified to determine the exact market value by inspecting it. They decide to delegate the search to a real-estate agent. The agent can have an individual valuation for each house (e.g., a provision that is paid internally for selling the house). This individual value might also depend on the condition of the house, which is revealed only after inspection. Similar situations arise, e.g., when the principal wants to invest in a financial product and delegates this search to a financial agent.

In these settings, as a consequence of delegating the search, \calP can merely observe the outcome of the search, not the other actions taken by \calA. This leads to a possibly undesirable outcome for \calP, even when they have the power to commit to accepting or rejecting certain outcomes~\cite{Bechtel22DelegatedPandorasBox}.

A natural approach to align the incentives of an agent with the goals of the principal in such hidden-action settings are \emph{contracts}. A contract allows for utility transfers to be made from the principal to the agent as a function of the outcome of the delegated task. While contract theory is a well-established~\cite{BoltonD05} and celebrated~\cite{Nobel16} area in economics, the interest in \emph{algorithmic} contract design has surged only relatively recently, see e.g.,~\cite{Duetting19SimpleVsOptimalContracts, DuttingRT21, Duetting21CombinatorialContracts, CastiglioniM022, GuruganeshSW023}.

In this work, we initiate the study of algorithmic contract design for exploration problems. Specifically, we consider the following problem (for a fully formal definition, see Section~\ref{sec:prelims}): There are $n$ boxes, the $i$-th of which contains a prize, drawn independently from a known probability distribution, and can be opened by paying a known exploration cost $c_i$. The $j$-th prize from box $i$ has some utilities $a_{ij}$ and $b_{ij}$ for \calA and \calP, respectively. Based on that information, \calP commits to a contract, specifying a transfer $t_{ij}$ for each prize, where we assume non-negative transfers (limited liability). Given such a contract, \calA picks an adaptive search strategy that allows to optimize the trade-off between value of the selected prize (for \calA) and total exploration cost. \calA uses the strategy to sequentially open boxes by paying the exploration cost. Depending on the observed prize, it determines the next box to open or to stop by selecting a prize from an open box. Assuming \calA selects prize $j$ from box $i$, the resulting final utility of \calP is $b_{ij}-t_{ij}$. The goal is to efficiently compute a contract that maximizes the expectation of the latter quantity. Here we assume that ties in the optimization of \calA are broken in favor of \calP, which is a standard assumption that makes the maximization problem well-defined.

For the problem of computing optimal contracts, there is a solution by linear programming that runs in time polynomial in the number of outcomes and the number of actions of the agent (cf., e.g.,~\cite{Duetting19SimpleVsOptimalContracts}). Note that this result does not yield a polynomial-time algorithm for our problem since the action space is extremely succinctly represented---its size is lower-bounded by the number of opening orders, $n!$.

\paragraph{Our Contribution.} We show that, in a variety of natural cases, optimal contracts can be computed in polynomial time (such as, e.g., optimal \emph{linear} contracts, or optimal contracts when the agent has no intrinsic value). While our results leave open the question whether computing optimal contracts efficiently is always possible in general, our insights suggest that doing so is a highly non-trivial task and represents one (of several) very interesting avenue for future research.

Notably, \calA faces an exploration problem known in the literature as the \emph{Pandora's Box Problem}. It was famously proposed and analyzed by Weitzman~\cite{Weitzman79OptimalSearch}, who characterized optimal search strategies by a simple greedy rule. Starting from~\cite{KleinbergWW16}, this problem has recently received a lot of attention at the intersection of economics and computer science (see, e.g.,~\cite{Singla18, GuptaJSS19, BeyhaghiK19, BoodaghiansFLL23, BeyhaghiC23, FuLL23} and a recent survey~\cite{BeyhaghiCSurvey23}).
We reformulate Weitzman's characterization for \calA's task: The \emph{fair cap} of box $i$ is a value such that the expected amount of $a_{ij}+t_{ij}$ exceeding the fair cap is precisely $c_i$. The boxes are considered in non-increasing order of fair caps. If the best $a_{ij}+t_{ij}$ (initially 0) observed so far is larger than the next fair cap, the next box is opened; if it is smaller, a corresponding prize is accepted. If neither is the case, \calA is indifferent.

The first observation we make is that Weitzman's policy does \emph{not} solve \calA's problem entirely. The reason is that \calA breaks ties in favor of \calP---recall that this is the standard assumption that makes the maximization problem well-defined. Note that the number of boxes that share some given fair cap $\varphi$ may be linear in $n$, leading to a number of possible ways of choosing the order that is exponential in $n$. Indeed, the order of such boxes may matter to \calP since different orders may lead to stopping with prizes that lead to vastly different utilities for \calP. A complicating factor is that ties whether to select a prize with agent utility equal to $\varphi$ have to be broken simultaneously. Interestingly, we show that the problem of breaking ties can be reduced---in polynomial time but in a non-obvious way---to solving \emph{different instances} of the classic Pandora's box problem.

With a solution of \calA's problem at hand, the main trade-off inherent to an optimal contract is as follows: Transfers manipulate the fair caps, the resulting exploration order, and the acceptance decisions of \calA in such a way that leads to a more favorable outcome for \calP\ -- but transfers are costly for \calP which is unfavorable. 

We start by considering \emph{linear contracts}~\cite{Holmstrom87AggregationAndLinearity,Duetting19SimpleVsOptimalContracts}, a class of simple and very intuitive contracts that has received a lot of attention due to its widespread use in applications. The agent receives a constant fraction of the principal's revenue as commission. Formally, a linear contract is represented by a number $\alpha\in[0,1]$, which defines transfers $t_{ij}=\alpha\cdot b_{ij}$.

To compute linear exploration contracts, we identify a polynomial-time computable set of \emph{critical} values of $\alpha$ at which the set of \calA-optimal policies (i.e., the set of policies that satisfy the above description) changes. By reasoning about the behavior of \calP's utility between any two such values, we can argue that the optimal contract needs to be a critical value. This type of argument is reminiscent of arguments from the recent literature (e.g.,~\cite{Duetting21CombinatorialContracts, Duetting24CombinatorialContractsBeyond, Dughmi24SupermodularContracts}), but here it requires special care. Due to the tie breaking in favor of \calP, which---as discussed above---is a key property here, we need to keep track of the \emph{set} of \calA-optimal policies.

For standard contract design, the authors in~\cite{Duetting19SimpleVsOptimalContracts} describe a class of instances in which the restriction to linear contracts leads to a multiplicative loss of $n$ in \calP's achievable utility, where $n$ is the number of actions.
In this class of instances, action $i$ with cost $c_i$ deterministically leads to an individual outcome with value $R_i$.
We can create one box for each action $i$, with opening cost $c_i$ and a single outcome with agent utility $0$ and principal utility $R_i$. This leads to a straightforward one-to-one correspondence between contracts in the two instances that preserves utilities and linearity.
Thus, the multiplicative loss of $n$ transfers to our model, which motivates our search for optimal \emph{general} exploration contracts.

We first consider general contracts under the standard assumption that $a_{ij}=0$, i.e., that the agent has no intrinsic motivation to explore any boxes and is only motivated by payments from \calP. 
In contract design, this assumption is usually without loss of generality, since arbitrary intrinsic agent values $a_{ij} \neq 0$ can be accounted for with action costs to obtain an equivalent instance with agent values $a_{ij} = 0$. 
However, for the exploration contracts we consider, such an adjustment changes the costs of \emph{exploration strategies} (rather than individual boxes). 
This substantially changes the structure of the problem. 
Indeed, as we see below, optimal policies for the problem with non-zero agent value have substantially different properties than the ones for zero value. 

Clearly, the utility that \calP can extract by solving the exploration problem themselves (i.e., ignoring the agent and paying the cost themselves) is a straightforward baseline; no policy can extract more value for \calP. Solving the problem from \calP's perspective yields fair caps $\smash{\varphi_i^\calP}$. If all $a_{ij}  = 0$, then by transferring any value exceeding $\smash{\varphi_i^\calP}$ to \calA, the \calP-optimal policy is adopted by \calA, without a loss in utility for \calP. 
This implies, in particular, that delegating the exploration to \calA using an optimal contract yields the same utility for \calP as if she performed the exploration on her own.
Formally, the idea is somewhat inspired by~\cite{KleinbergWW16} for the original Pandora's Box problem but different.

An orthogonal special case that has been considered in the literature (see, e.g.,~\cite{Bechtel22DelegatedPandorasBox}) is that of \emph{binary boxes}. That is, each box contains one of only two prizes, and one of them has $0$ value for both \calP and \calA while the other one is arbitrary. Without any payments, the fair cap (from \calA's perspective again) of each box is its \emph{basic fair cap}. The principal now faces the question of whether to lift (by transfers) the basic fair caps of favorable boxes above (or up to) the fair caps of less favorable boxes.
While the latter is true in general, the special structure of binary boxes allows us to compute an optimal exploration contract by considering boxes in non-increasing order of their basic fair caps and moving their fair cap \emph{greedily}. 
Note that this algorithm alone does not suffice to tackle more general cases. Indeed, when boxes are not binary, it may be profitable for the principal to choose different payments for two boxes with the same distribution and the same fair caps -- a condition we encounter in the next and final case.

The final case we consider has a more general value structure, but we compromise on the boxes having different distributions. Specifically, we consider instances with only a single positive prize for \calP, and all distributions are identical. Still, transfers can be made depending on the box from which the prize originates. Again, the complexity is reduced by the fact that only a single payment per box has to be decided. Through an intricate sequence of exchange arguments and the help of continuization, we establish that there is an optimal contract with a simple structure: (i) In a first phase, the prize with positive principal value gets immediately accepted, and all transfers are identical; (ii) all transfers for boxes in the second phase are also identical. These conditions allow to find the optimal contract by enumerating a polynomial-sized set of contracts.

In this final case, we also give a family of instances in which both the above phases exist and have substantial lengths in any optimal solution. Hence, perhaps surprisingly, such a relatively simple case already has optimal contracts with a fairly complicated structure.

\paragraph{Overview.} After a review of related work in the subsequent section, we formally introduce our problem in Section~\ref{sec:prelims} along with preliminary observations. Section~\ref{sec:linear} treats linear contracts. Section~\ref{sec:general} deals with general contracts in several special cases (no agent value in Section~\ref{sec:noValue} and binary boxes in Section~\ref{sec:binary}). General contracts for i.i.d.\ boxes with a single positive prize are discussed in Section~\ref{sec:iid}. Missing proofs are provided in Appendices~\ref{app:tieBreak},~\ref{app:binary} and~\ref{app:iid}.

\subsection{Related Work}

Contract design has recently gained a lot of attention from a computational point of view.
Duetting et al.~\cite{Duetting19SimpleVsOptimalContracts} advocated the study of linear contracts as an alternative to the more complicated (and sometimes unintuitive) general contracts. They show robustness results and give parameterized approximation guarantees w.r.t.\ optimal (general) contracts. For settings in which the outcome space is succinctly represented, hardness results are shown in~\cite{DuttingRT21}. In a similar but different approach, \emph{combinatorial contracts}~\cite{Duetting21CombinatorialContracts} let the agent choose a subset of actions which stochastically determines the (binary) outcome. Important recent work in the area of combinatorial contracts includes \cite{Duetting24CombinatorialContractsBeyond}, \cite{Dughmi24SupermodularContracts}. Contracts were also studied with multiple agents~\cite{Duetting23MultiAgentContracts, Dughmi24SupermodularContracts}, ambiguity~\cite{Duetting23AmbiguousContracts}, and private types~\cite{Alon21ContractsWithPrivateCosts} with connections to classic mechanism design~\cite{AlonDLT23}. More generally, contracts have been of interest in specific application domains, such as classification in machine learning~\cite{SaigTR23}.

Most closely related is prior work by Postl~\cite{Postl04}, who analyzes our contract-design problem for two boxes without intrinsic agent value. The paper provides an explicit characterization of the chosen box in the optimal contract, but neither considers algorithmic aspects nor the reduction to the standard Pandora's Box problem we discover in Sec.~\ref{sec:noValue}. Concurrent to our work, Ezra et al~\cite{EzraFS24} study the special case without agent value. In this special case, they also give a polynomial-time algorithm for linear contracts (without discussing the issues of tie-breaking in favor of \calP). For general contracts, the paper studies a technical approach when there is a constant number of prizes. Finally, the authors show NP-hardness for correlated distributions in the boxes.

Delegation is a related approach in the principal--agent framework and also received a lot of attention in the economics literature starting from seminal work of Holstrom~\cite{Holmstrom79MoralHazard}. Computational aspects of delegation are receiving interest recently, initiated by Kleinberg and Kleinberg~\cite{KleinbergKleinberg18}. In these models, the principal \calP delegates a (search) task to an agent \calA, again with potentially different interests, who has to inspect alternatives and \emph{propose} an observed prize to \calP.
Instead of committing to a contract with payments, \calP limits the set of prizes she is willing to accept upon proposal by \calA. 
A prominent objective is to find good acceptance policy for \calP and bound their performance by the \textit{delegation gap}, which measures the multiplicative loss in utility for \calP in comparison to the case when \calP performs the (undelegated) search problem themselves. 
On a technical level, there are close connections to prophet inequalities~\cite{KleinbergKleinberg18} and online contention resolution schemes, which were established in different variants, such as multiple agents~\cite{Hajiaghayi23MultipleAgents}, stochastic probing~\cite{Bechtel21DelegatedStochasticProbing, Bechtel22DelegatedPandorasBox}, as an online problem~\cite{Braun23DelegatedOnlineSearch}, or with limited information about agent utilities~\cite{CastiglioniM021}.
Perhaps most related is a version of the problem studied in~\cite{Bechtel22DelegatedPandorasBox}. Here, a principal delegates an instance of the Pandora's Box problem by committing to a set of acceptable outcomes. An agent has to perform the costly inspection of alternatives. This standard delegation setting has some obvious limitations, e.g., the agent would never perform an inspection if his expected intrinsic value does not cover the inspection cost. A model variant with transfers is also studied in~\cite{Bechtel22DelegatedPandorasBox} in which the principal can reimburse the agent for exploration costs. Crucially, this is impossible in \emph{hidden-action principal-agent} settings that we consider, where \calP simply cannot directly reimburse \calA for his expenses, since she cannot observe which inspections have been performed by \calA.

In a more general context, online optimization in principal--agent settings has generated significant interest recently, e.g., in game-theoretic versions of the Pandora box problem~\cite{DingFHTX23}, general optimal stopping problems~\cite{HahnHS20, HahnHS22}, or with unknown agent utilities~\cite{CastiglioniCM020,ZuIX21,FengTX22,BabichenkoTXZ22}.

\section{Preliminaries}
\label{sec:prelims}

There are two parties, a principal $\calP$ and an agent $\calA$, and $n$ boxes. Each box $i \in [n]$ has an opening cost $c_i \ge 0$ and contains one of $m$ possible prizes. Prize $j \in [m]$ in box $i$ yields a value-pair $(a_{ij},b_{ij})$, where $a_{ij} \ge 0$ is the value for $\calA$ and $b_{ij} \ge 0$ the value for $\calP$. For each box $i \in [n]$, the prize in the box is distributed independently according to a known prior. We denote by $p_{ij}\geq 0$ the probability that prize $j$ is in box $i$.

$\calP$ would like to motivate $\calA$ to open and inspect boxes in order to find and select a good prize (for $\calP$). $\calP$ cannot observe the actions taken by $\calA$ (i.e., which boxes were opened in which order), only the final selected box with the prize inside is revealed. Instead, $\calP$ can commit to an \emph{exploration contract} (or simply, \emph{contract}) $T$: For each selected prize $(i,j)$, she specifies an amount $t_{ij} \in [0,b_{ij}]$\footnote{We make the natural assumption that $t_{ij} \le b_{ij}$, i.e., every transfer in the contract is bounded by the prize value of $\calP$. For all scenarios we study here (linear contracts, no agent value, binary boxes, i.i.d.\ with single positive prize for $\calP$) it is straightforward to see that this assumption is w.l.o.g. -- there is an optimal contract that satisfies $t_{ij} \le b_{ij}$ for all $(i,j) \in [n]\times [m]$. It is a very interesting open problem to prove that this holds beyond the cases we study here.} that she pays to $\calA$. Then the final utility of $\calP$ is $b_{ij} - t_{ij}$. 

Our model is an instance of the standard principal-agent model, where the hidden action of the agent $\calA$ is the policy to inspect and select from the boxes. Specifically, the task of $\calA$ is to open and inspect boxes to find and select at most one prize from an opened box. Formally, a \emph{policy} specifies in every situation, based on the previously seen prizes, which box to open next, if any, or otherwise which prize to select from an opened box. Given a contract $T$, she is facing the classic Pandora's Box problem~\cite{Weitzman79OptimalSearch}: She obtains a utility of $a_{ij} + t_{ij}$ for the selected prize ($0$ if no prize is selected) minus the opening cost $c_i$ for all opened boxes. Weitzman~\cite{Weitzman79OptimalSearch} shows that every optimal policy of $\calA$ that maximizes his expected total utility (called \emph{\calA-optimal policy} in the following) is in the form of the following index policy: For each box $i \in [n]$ compute an index or \emph{fair cap} $\faircap{i}$ such that 
\[
    \sum_{j \in [m]} p_{ij} \max\{0, a_{ij} + t_{ij} - \faircap{i}\} = c_i.
\]
The boxes are considered in non-increasing order of fair caps. Suppose box $i$ is the next unopened box in this order, and the best prize \calA has found thus far is $v$. If no box has been opened, then $v=0$, and otherwise $v=a_{i'j'}+t_{i'j'}$ for some $(i',j')$ where $i'$ is a box that has already been opened, and prize $j'$ has been found in it. There are three cases: (1) $v<\faircap{i}$: \calA opens box $i$; (2) $v=\faircap{i}$: \calA is indifferent between opening box $i$ and stopping opening boxes; and (3) $v>\faircap{i}$: \calA stops opening boxes. Note that any box $i$ with $\faircap{i} < 0$ is never opened. For each $(i,j) \in [n]\times [m]$ we define the capped value $\kappa_i = \min\{ a_{ij} + t_{ij}, \faircap{i}\}$. A policy as described above always selects the prize in box $i^* \in \arg\max_i \kappa_i$, i.e., the prize with the largest capped agent value.

Even under these conditions, there may still be several \calA-optimal policies. Specifically,
\begin{enumerate}
    \item[(i)] the choice of $\faircap{i}$ is not unique when $c_i=0$,
    \item[(ii)] the non-decreasing order of fair caps may not be unique,
    \item[(iii)] in the case $v=\faircap{i}$ above, the choice of whether to stop or continue is not unique,
    \item[(iv)] among the observed prizes, the prize $(i,j)$ maximizing $a_{ij}+t_{ij}$ may not be unique.
\end{enumerate}

We assume that \calA breaks ties in favor of $\calP$\footnote{This is a standard assumption in bi-level problems. On a technical level, it ensures that the optimization problem of finding an optimal contract for $\calP$ is well-defined.}. Among the set of $\calA$-optimal policies, she selects a \emph{\calP-optimal} one, i.e., a policy that maximizes the expected utility for \calP. (Note that such a policy may be vastly different from an \calP-optimal policy among \emph{all} policies.)
We call a policy that is selected in this way an \emph{optimal policy (under contract $T$)}.

In case (iv), it is clear how to resolve the ambiguity in favor of $\calP$: Simply select $(i,j)$ maximizing $b_{ij}-t_{ij}$. The other cases, however, can give rise to a very large number of potential policies. Interestingly, in Appendix~\ref{app:tieBreak} we show that it is always possible to implement the tie-breaking in polynomial time and find an optimal policy under contract $T$. We prove this result through a reduction to the original Pandora's Box problem.

\begin{theorem}\label{thm:best-order}
    Given a contract $T$, an optimal policy can be computed in polynomial time.
\end{theorem}

We study contracts for $\calP$ that steer the index policy executed by $\calA$ towards good outcomes for $\calP$.
We explore both linear and general contracts. A \emph{linear} contract is given by a single number $\alpha \in [0,1]$, and $t_{ij} = \alpha \cdot b_{ij}$. 
Linear contracts are popular because of their simplicity, but they often suffer from substantial limitations of the achievable revenue. 
As such, we also explore \emph{general} contracts, in which we only require every payment to be non-negative $t_{ij} \ge 0$ for every $i \in [n]$ and $j \in [m]$. A (linear) contract $T$ that, among all possible (linear) contracts, achieves maximum expected utility for \calP when \calA executes an optimal policy under $T$, is called \emph{optimal (linear) contract}.

\section{Optimal Linear Contracts}
\label{sec:linear}

In this section we consider linear contracts. Recall that these are contracts characterized by a single $\alpha\in[0,1]$ such that $t_{ij}=\alpha\cdot b_{ij}$ for all $(i,j)\in[n]\times[m]$. Our result is the following.

\begin{theorem}\label{thm:linear}
    An optimal linear exploration contract can be computed in polynomial time.
\end{theorem}

Let $\alpha_1,\alpha_2\in[0,1]$ such that the \emph{set of \calA-optimal} policies (i.e., index policies) is the same for all $\alpha\in[\alpha_1,\alpha_2]$. Consider the expected utility of \calP according to an optimal policy (i.e., an \calP-optimal policy among this set of \calA-optimal policies) as a function of $\alpha$ within the interval $[\alpha_1,\alpha_2]$. The following straightforward observation implies that this function is a linear and non-increasing function of $\alpha$.

\begin{observation}\label{obs:utility-principal}
    For every policy, there exists a constant $c$ such that the expected utility of \xspace\calP as a function of $\alpha$ is $(1-\alpha)\cdot c$.
\end{observation}

Hence, within $[\alpha_1,\alpha_2]$, an optimal contract is $\alpha_1$. To find a polynomial-time algorithm for computing the global optimum $\alpha$, it therefore suffices to show that, in polynomial time, one can find a partition of the interval $[0,1]$ into (polynomially many) subintervals such that, in each subinterval, the set of \calA-optimal policies is constant. (Strictly speaking, we will encounter a technicality at the endpoints of the subintervals.) In the following, we will define a (polynomial-time computable) set of \emph{critical values} in $[0,1]$ such that, in an interval (strictly) between any two consecutive critical values, the set of \calA-optimal policies is constant.

Towards this definition, for any given $\alpha \in [0,1]$ and every $(i,j) \in [n] \times [m]$, let $v_{ij}(\alpha) := a_{ij} + \alpha \cdot b_{ij}$ denote the value of prize $j$ in box $i$ for $\calA$.
Furthermore, let $\faircap{i}(\alpha)$ denote the (unique) fair cap of box $i \in [n]$ with $c_i > 0$ as a function of $\alpha$, i.e., it holds that
\[\sum_{j \in [m]} p_{ij} \max\left\{0, v_{ij}(\alpha) - \faircap{i}(\alpha)\right\} = c_i\] for all $\alpha \in [0,1]$. Note that $\faircap{i}(\alpha)$ is a continuous function of $\alpha$ in $[0,1]$.

We call $\alpha\in[0,1]$ a critical value if $\alpha\in\{0,1\}$ or one of the following properties holds:
\begin{enumerate}
    \item[(a)] There exist boxes $i,i'\in[n]$ with $c_i>0,c_{i'}>0$ such that the order of their fair caps changes at $\alpha$. Formally, $\faircap{i}(\alpha)=\faircap{i'}(\alpha)$, and there exists $\varepsilon>0$ such that $\faircap{i}(\alpha')\neq\faircap{i'}(\alpha')$ for all $\alpha'\in(\alpha-\varepsilon,\alpha)$ or for all $\alpha'\in(\alpha,\alpha+\varepsilon)$.

    \item[(b)] There exist $i, i' \in [n]$ with $c_i>0,c_{i'}>0$ and $j \in [m]$ such that the order between the value of prize $(i,j)$ for \calA and the fair cap of box $i'$ changes. Formally, $v_{ij}(\alpha) = \faircap{i'}(\alpha)$, and there exists $\varepsilon>0$ such that $v_{ij}(\alpha') \neq \faircap{i'}(\alpha')$ for all $\alpha'\in(\alpha-\varepsilon,\alpha)$ or for all $\alpha'\in(\alpha,\alpha+\varepsilon)$.

    \item[(c)] There exists $i\in[n]$ with $c_i>0$ such that $\faircap{i}(\alpha)=0$.
\end{enumerate}
Observe that the critical values indeed have the property that, for any two such values $\alpha^c_1,\alpha^c_2$, the set of \calA-optimal policies is constant within $(\alpha^c_1,\alpha^c_2)$.
Note that any \calA-optimal policy is also \calA-optimal under both $\alpha^c_1$ and $\alpha^c_2$, but there are potentially additional \calA-optimal policies under $\alpha^c_1$ or under $\alpha^c_2$. Together with Observation~\ref{obs:utility-principal}, this implies that the optimal contract \emph{is} a critical value.

The following auxiliary lemma supports the analysis of the number of critical values.

\begin{lemma}
    \label{lem:opt-linear:faircap-piecewise-linear}
    For each $i \in [n]$ with $c_i > 0$, the fair cap $\faircap{i} : [0, 1] \to \mathbb R$ is a monotone convex piece-wise linear function with at most $2m+1$ linear segments.
\end{lemma}
\begin{proof}
    For any $\alpha \in [0, 1]$, it holds $\sum_{j \in [m]} p_{ij} \max\{0, v_{ij}(\alpha) - \faircap{i}(\alpha)\} = c_i$ by definition.
    An equivalent formulation would be
    \[
        \faircap{i}(\alpha) = \frac {\sum_{j \in S_i(\alpha)} p_{ij} v_{ij}(\alpha) - c_i} {\sum_{j \in S_i(\alpha)} p_{ij}},
    \]
    where $S_i(\alpha) := \{j \in [m] : v_{ij}(\alpha) \ge \faircap{i}(\alpha)\}$.
    Hence, $\faircap{i}(\alpha)$ is the weighted average (minus offset $c_i$) of all $v_{ij}(\alpha)$ that are in $S_i(\alpha)$, i.e., above $\faircap{i}(\alpha)$.
    For a fixed set $S_i(\alpha)$, this means that the fair cap $\faircap{i}$ is linear in $\alpha$, as every $v_{ij}$ is linear in $\alpha$.
    Therefore, the slope of $\faircap{i}$ only changes at intersections with some $v_{ij}$, where the set $S_i(\alpha)$ changes.
    We argue now that the fair cap can only increase at those intersections (cf.\ Fig.~\ref{fig:slope-of-faircap-only-increases}).
    
    For every $(i,j) \in [n] \times [m]$ and $\alpha \in (0, 1]$ with $v_{ij}(\alpha) = \faircap{i}(\alpha)$, but $v_{ij}(\alpha-\varepsilon) \neq \faircap{i}(\alpha-\varepsilon)$ for some $\varepsilon > 0$ , there are two cases:
    \begin{enumerate}
        \item $v_{ij}(\alpha-\varepsilon) > \faircap{i}(\alpha-\varepsilon)$.
        Then $j \in S_i(\alpha - \varepsilon)$, but $j \notin S_i(\alpha + \varepsilon')$ for some $\varepsilon' > 0$.
        Hence the slope of $\faircap{i}$ was greater than the slope of $v_{ij}$ before the intersection at $\alpha$.
        The new weighted average of affine functions that are above the fair cap can only increase when $v_{ij}$ is not contributing to that average anymore.
        \item $v_{ij}(\alpha-\varepsilon) < \faircap{i}(\alpha-\varepsilon)$.
        Then $j \notin S_i(\alpha - \varepsilon)$, but $j \in S_i(\alpha)$. 
        Hence the slope of $\faircap{i}$ was less than the slope of $v_{ij}$ before the intersection at $\alpha$.
        The new weighted average of affine functions that are above the fair cap can only increase when $v_{ij}$ is starting to contribute to that average.
    \end{enumerate}
    Thus the slope of $\faircap{i}(\alpha)$ is never decreasing when $\alpha$ increases, which makes $\faircap{i}$ convex.
    As $v_{ij}$ is an affine function, every $v_{ij}$ can intersect at most twice with $\faircap{i}$.
    Therefore, there are at most $2m+1$ linear segments of $\faircap{i}$ on the interval $[0, 1]$.
\end{proof}

    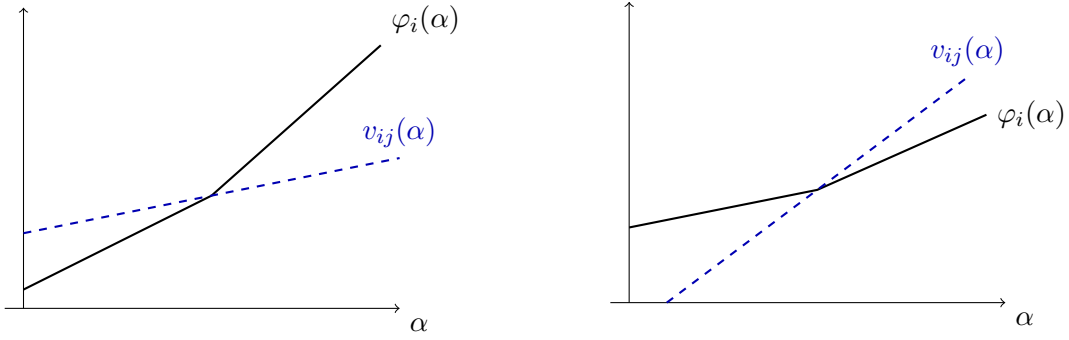
\begin{figure}[t]
		\centering
		\begin{minipage}{0.48\textwidth}
			\centering
			\begin{tikzpicture}
				\draw[->] (-0.25,0) -- (5,0) node[below right] {$\alpha$};
				\draw[->] (0,0) -- (0,4);
				
				\draw[-, thick] (0,0.25) -- (2.5,1.5) -- (4.75,3.5) node[above right] {$\faircap{i}(\alpha)$};
				\draw[dashed, thick, blue!70!black] (0,1) -- (5,2) node[above] {$v_{ij}(\alpha)$};
				
			\end{tikzpicture}
		\end{minipage}
		\begin{minipage}{0.48\textwidth}
			\centering
			\begin{tikzpicture}
				\draw[->] (-0.25,0) -- (5,0) node[below right] {$\alpha$};
				\draw[->] (0,0) -- (0,4);
				
				\draw[-, thick] (0,1) -- (2.5,1.5) -- (4.75,2.5) node[right] {$\faircap{i}(\alpha)$};
				\draw[dashed, thick, blue!70!black] (0.5,0) -- (4.5,3) node[above] {$v_{ij}(\alpha)$};
				
			\end{tikzpicture}

		\end{minipage}
		\caption{
        Two cases for an intersection between $v_{ij}$ and $\faircap{i}$.
        Left: The slope of $v_{ij}$ was less than slope of $\faircap{i}$ before the intersection.
        By definition, it contributed to the weighted average slope that defines the slope of $\faircap{i}$.
        After the intersection, it does not contribute anymore, and the weighted average only increases.
        Right: The slope of $v_{ij}$ was greater than slope of $\faircap{i}$ before the intersection.
        Symmetrical arguments.
        }
		\label{fig:slope-of-faircap-only-increases}
	\end{figure}
 
According to Lemma~\ref{lem:opt-linear:faircap-piecewise-linear}, there are at most $\calO(n)$ critical values induced by case (c). By the same lemma, there are at most $\calO(nm)$ critical values for every box $i \in [n]$ with $c_i > 0$ induced by case (b), as $\faircap{i}$ is convex and there are $nm$ affine functions $v_{i'j}$, each of which intersects $\faircap{i}$ at most twice. Thus, $\calO(n^2m)$ critical values in total are due to cases (b) and (c).
Similarly, each linear segment of $\faircap{i}$ can intersect with another convex function $\faircap{i'}$ at most twice.
As there are at $\calO(m)$ linear segments for $\alpha \in [0,1]$ and $n$ possible functions $\faircap{i'}$, there are at most $\calO(nm)$ such intersections for $\faircap{i}$.
Thus, case (a) induces at most $\calO(n^2m)$ critical values.
Overall, the number of critical values is polynomial in $n$ and $m$, and the set of them can be computed in polynomial time.

By Theorem~\ref{thm:best-order}, we can compute an optimal contract for any given critical value in polynomial time. The expected utility of \calP under this contract can also be computed in polynomial time.
Thus, we can enumerate all critical values and take the best contract with respect to the expected utility for $\calP$, implying Theorem~\ref{thm:linear}.

\section{Optimal General Contracts}
\label{sec:general}

\subsection{No Intrinsic Agent Value}
\label{sec:noValue}

In the consideration of contract problems, it is often assumed that the agent has no intrinsic value (only cost) and receives benefit only via transfers from the principal. In this case, we have $a_{ij} = 0$ for all $(i,j) \in [n] \times [m]$. We show that under this assumption, we can compute an optimal contract in polynomial time.

\begin{theorem}\label{thm:no-intrinsic}
    Suppose that $a_{ij} = 0$ holds for all $(i,j) \in [n] \times [m]$. Then an optimal exploration contract can be computed in polynomial time.
\end{theorem}

Without any transfers from $\calP$, $\calA$ has no intrinsic motivation to open any box -- except the ones with inspection cost 0. Specifically, if the agent opens box $i \in [n]$, the (expected) payments from $\calP$ have to cover the inspection costs for $\calA$, i.e., it holds that $\sum_{j \in [m]} p_{ij} t_{ij} \ge c_i$. As a direct consequence, $\calP$ cannot obtain more utility from a contract with $\calA$ over the one obtained by exploring boxes and paying inspection costs by herself (she can simply imitate $\calA$'s behavior under the contract). 

\begin{definition}
    A given contract $T$ \emph{implements} a policy $\pi$ if (1) $\pi$ is optimal under $T$, and (2) $\sum_{j \in [m]} p_{ij} t_{ij} = c_i$ for all $i \in [n]$.
\end{definition}

As observed above, an optimal contract cannot yield more utility for $\calP$ than an optimal policy $\pi^*$ that she can apply by herself (paying the costs herself). Hence, the next lemma implies Theorem~\ref{thm:no-intrinsic}.

\begin{lemma}
    There exists a contract $T^*$ that implements $\pi^*$.
\end{lemma}
\begin{proof}
    Recall that the following is an optimal policy $\pi^*$ for $\calP$ (when she pays the costs herself). Define fair caps $\faircapP{i}$ for each box $i \in [n]$ such that $\sum_{j \in [m]} p_{ij} \max\{0, b_{ij} - \faircapP{i}\} = c_i$. Then $\calP$ opens boxes in non-increasing order of their fair caps as long as the best prize observed so far does not exceed the fair cap of the next box.

A contract $T^*$ that implements $\pi^*$ is given by payments $t_{ij} := \max\left\{0, b_{ij} - \faircapP{i}\right\}$ for all $(i,j) \in [n] \times [m]$.
Clearly, we have
\begin{equation}
    \label{eq:payments-equal-cost}
    \sum_{j \in [m]} p_{ij} t_{ij} = c_i, \hspace{1cm} \text{for all $i \in [n]$.}
\end{equation}
Consequently, $\calA$ is faced with an instance of the Pandora's Box problem where her expected prize in each box $i$ equals the opening cost. Thus, if $\calA$ applies the (\calA-optimal) index policy for the emerging instance, the fair caps for $\calA$ are given by $\faircapA{i} = 0$ for all $i \in [n]$.

This means that $\calA$ is entirely indifferent about the opening order. As long as there are only prizes with $t_{ij} = 0$, $\calA$ is also indifferent about stopping to open further boxes and selecting any previously observed prize at any point in time. Thus, $\calA$'s behavior under these conditions can be assumed to be consistent with the one imposed by $\pi^*$. 

Restrictions only arise from the fact that $\calA$ \emph{must stop immediately} whenever a prize $j$ with $t_{ij} > 0$ is drawn from box $i$, according to $\calA$'s index policy. However, this is consistent with the behavior of the index policy $\pi^*$ for $\calP$: Note that $t_{ij} > 0$ if and only if $b_{ij} > \faircapP{i}$. Hence, $\calP$ also stops when prize $j$ is drawn from box $i$, because $b_{ij}$ exceeds the fair cap of box $i$ and, thus, the fair cap of the next box in the order.

Therefore, $\pi^*$ is an optimal policy under contract $T^*$. Together with~\eqref{eq:payments-equal-cost} this shows that $T^*$ implements $\pi^*$. 
\end{proof}

\subsection{Binary Boxes}
\label{sec:binary}

We study a subclass of the problem where every box $i \in [n]$ contains two prizes with positive probability.
One of those prizes has value 0 for both $\calP$ and $\calA$, and is called \emph{0-prize}.
The other prize (the \emph{positive prize}) is denoted by the value pair $(a_i, b_i)$ with $a_i, b_i \ge 0$ and has probability $p_i \in (0,1]$.
Consequently, the 0-prize is drawn from box $i$ with probability $1 - p_i$.
Note that there cannot be a positive payment for the 0-prize. 
Thus, payments for box $i$ are given by a single $t_i \ge 0$, the payment for the positive prize.

Depending on a payment of $t \ge 0$, we define the fair cap $\faircap{i}(t)$ of box $i$ for $\calA$ by
\begin{equation}
    \label{eq:fairCap}
    \faircap{i}(t) = t + a_i - \frac {c_i} {p_i}.
\end{equation}

\begin{lemma}\label{lem:binary-boxes:retrieve-policy-from-contract}
    Given a contract $T$, there is an optimal policy that has the following properties. \calA opens boxes in order of the fair cap defined by \eqref{eq:fairCap}. Boxes with the same fair cap are ordered according to the value $b_i - t_i$. \calA accepts the first box with a positive prize (if any).
\end{lemma}

\begin{proof}
As discussed in Section~\ref{sec:prelims}, the fair cap for boxes is unique whenever $c_i > 0$. For boxes with $c_i = 0$ we also assume that the fair cap is given as in \eqref{eq:fairCap}. This is the smallest fair cap for this box and defers the inspection of box $i$ to the latest point. Clearly, the agent is indifferent regarding this choice. Consequently, it is $\calA$-optimal to stop immediately and accept as soon as $\calA$ opens a box with a positive prize in it, since the value is $a_i + t_i \ge \faircap{i}(t_i) \ge \faircap{i+1}(t_{i+1})$. 

To argue that this behavior is also $\calP$-optimal, suppose some box $i^*$ with $c_{i^*} = 0$ gets a higher fair cap $\faircap{i^*}(t_{i^*}) > a_{i^*} + t_{i^*}$ and gets opened earlier (in accordance with the resulting index policy of \calA). With such fair cap for $i^*$, $\calA$ does not stop until all subsequent boxes $i$ with fair cap $\faircap{i}(t_{i}) > a_{i^*} + t_{i^*}$ are opened. As such, we can defer the opening of box $i^*$ until after these boxes are opened. Now if there are several boxes with the same fair caps, it is optimal for \calP if the opening is done in non-increasing order of $b_i - t_i$. This guarantees that there is no subsequent box with the same fair cap and a prize that has a better value for $\calP$. Conversely, any box with the same fair cap and potentially better value for $\calP$ has been inspected before. As such, using minimal fair caps for boxes with $c_i=0$, breaking ties w.r.t.\ $b_i-t_i$, and stopping as soon as a positive prize is observed is also $\calP$-optimal.
\end{proof}

Let us define a \emph{basic} fair cap of box $i$ to be $\faircap{i}(0)$. A box with negative basic fair cap would not be considered by $\calA$ unless the payment $t_i$ is at least $c_i/p_i - a_i$, in which case the fair cap would be precisely $0$. Note that boxes with prohibitively high cost $p_i(a_i + b_i) \le c_i$ are w.l.o.g.\ never opened by $\calA$ under any contract $T$. In what follows, we exclude such boxes from consideration. For the remaining boxes $p_i (a_i + b_i) > c_i$, and, hence, $b_i > c_i/p_i - a_i =: \tilde t_i$. Whenever $\tilde t_i > 0$, the exploration cost exceeds the expected value for \calA in box $i$. As such, a transfer of $\tilde t_i$ is clearly necessary to motivate \calA to inspect box $i$ at all. 
We renormalize the box by adjusting the value of the positive prize to $(a_i + \tilde t_i, b_i - \tilde t_i)$. Then the basic fair cap becomes $\faircap{i}(0) = 0$.
Indeed, assuming $t_i\geq \tilde t_i$ is w.l.o.g.: Consider any contract $T$ in which $t_i<\tilde t_i$ for all $i$ in some non-empty set of boxes $B$, and the optimal policy $\pi$ under $T$. Note that choosing a contract $T'$ where $t_i$ is increased to $\tilde t_i$ for such boxes does not hurt: This will increase the fair caps in $B$ to 0. An \calA-optimal policy $\pi'$ for $T'$ can be obtained from $\pi$ by opening the boxes $B$ in the end if no prize has been found yet, in any order. The resulting utility for \calP in such cases is non-negative rather than 0 (as $\tilde{t}_i \leq b_i$); the utility in all other cases does not change.

In the remainder of the section, we consider the instance with renormalized boxes, i.e., each box $i$ has a basic fair cap $\faircap{i}(0) \ge 0$. Furthermore, we re-number the boxes and assume that the indices are assigned in non-increasing order of basic fair caps, i.e. $i < j \implies \faircap{i}(0) \ge \faircap{j}(0)$ for all $i, j \in [n]$. We break ties in the ordering w.r.t.\ non-increasing value of $b_i$.

Our next observation will allow us to restrict attention to the permutation and the fair caps in the policy of $\calA$.
By~\eqref{eq:fairCap}, fair caps are linear in payments. Clearly, for an optimal contract we strive to set smallest payments (and, hence, smallest fair caps) to induce a behavior of $\calA$.
For any given contract $T$, we say $T$ \emph{implements an ordering $\sigma$ of boxes} if there is an optimal policy under $T$ that considers boxes in the order of $\sigma$. For the reverse direction, we need a slightly more technical definition.

\begin{definition}\label{def:binary-boxes:pointwise-smallest-payments-contract}
For any given ordering $\sigma$ of boxes, we define a contract $T(\sigma)$ as follows:
For each position $i$ (occupied with box $\sigma(i)$), let $j = \max( \arg \max_{j' \in \{i, i+1, \dots, n\}} \{\faircap{\sigma(j')}(0)\})$ be the latest position of any subsequent box (including $\sigma(i)$ itself) with the highest basic fair cap of the subsequent boxes. In $T(\sigma)$, we assign payments to the unique positive prizes for the boxes on positions $i, i+1, \dots, j$ such that they all have a fair cap of exactly $\faircap{\sigma(j)}(0)$. If $T(\sigma)$ has feasible payments and implements $\sigma$, we call $T(\sigma)$ the \emph{canonical contract} for $\sigma$.
\end{definition}

We have the following lemma concerning whether $T(\sigma)$ is the canonical contract for $\sigma$.

\begin{lemma}\label{lem:binary-boxes:opt-order-implies-contract}
For any ordering $\sigma$, we can decide in polynomial time if $T(\sigma)$ is the canonical contract for $\sigma$ (and compute it in this case).
If $T(\sigma)$ is not the canonical contract for $\sigma$, then every contract $T$ that implements $\sigma$ is suboptimal for the given instance.
Otherwise, the canonical contract has point-wise minimal payments among all contracts that implement $\sigma$.
\end{lemma}

\begin{proof}
    The construction of $T(\sigma)$ and the subsequent feasibility check described in Definition~\ref{def:binary-boxes:pointwise-smallest-payments-contract} can be done in polynomial time. 
    Note that if $T(\sigma)$ is the canonical contract for $\sigma$, it uses point-wise minimal payments for implementing $\sigma$:
    Having less payments for any box would directly contradict implementation of $\sigma$, since every optimal policy considers boxes in the order of their fair caps, which are lower bounded by the respective basic fair caps.
    Thus, if $T(\sigma)$ is not the canonical contract for $\sigma$ due to infeasible payments, then implementing $\sigma$ is impossible with any contract, because even higher payments would be necessary.
    If $T(\sigma)$ is not the canonical contract for $\sigma$, but the construction in Definition~\ref{def:binary-boxes:pointwise-smallest-payments-contract} defines feasible payments, it defines a canonical contract $T(\sigma')$ of some other ordering $\sigma' \neq \sigma$ using point-wise minimal payments.
    Note that $\sigma'$ instead of $\sigma$ only emerges because of tie-breaking in favor of \calP, which implies that boxes with identical fair caps have to be considered in non-increasing order of $b_i - t_i$.
    Thus it could be possible to implement $\sigma$ with another contract $T$, but additional payments would be required to remove (some of the) ties.
    However, since tie-breaking was already done in favor of \calP, this contract $T$ clearly yields less utility than $T(\sigma')$ and hence is not an optimal contract for the given instance.
\end{proof}

In the following, we will discuss how to compute the fair caps for the policy of $\calA$ in an optimal contract. 
By Lemma~\ref{lem:binary-boxes:opt-order-implies-contract}, the order $\sigma$ that we compute in this way has to be implemented by $T(\sigma)$.

Let $e_{\calP}(\faircap{1}, \dots, \faircap{n})$ denote the expected utility for $\calP$ under a contract given by fair caps $\faircap{1}, \dots, \faircap{n}$. Algorithm~\ref{alg:binary-boxes} computes optimal fair caps for all boxes. Initially, the fair cap of each box is given by its basic fair cap. Then we calculate the optimal fair cap of boxes $1, \dots, n$ iteratively. In iteration $k$ of the outer for-loop, we find the optimal fair cap for box $k$ among the fair caps of previous boxes. Crucially, it is sufficient to test which of those fair caps yield maximum utility, even if the final fair caps of subsequent boxes are not determined yet. Note that by~\eqref{eq:fairCap}, to raise the fair cap of any box $i \in [n]$ to $\varphi > \faircap{i}(0)$, the required payment $t_i$ is a constant independent of the probabilities and values in box $i$ -- more precisely, it is exactly $t_i = \varphi - \faircap{i}(0)$.
\begin{algorithm}[t]
    \caption{Optimal Contract for Binary Boxes.}
    \label{alg:binary-boxes}
    \KwData{Binary boxes $1, \dots, n$ (non-increasing $\faircap{i}(0)$, tie-break: non-increasing $b_i$)}
    \KwResult{Contract $(t_1, \dots, t_n)$}

    $x_k \gets \faircap{k}(0)$ for all $k \in [n]$ \tcp*{Initialize with basic fair caps.}

    \For{$k = 1, \ldots, n$}{
        \For{$j = 1, \ldots, k$} {
             \tcc{Compute expected utility for $\calP$ when box $k$ had fair cap $x_j$}
            $\rho \gets x_j$ \;
            $e_j \gets e_{\calP}(x_1, \dots, x_{k-1}, \rho, x_{k+1}, \dots, x_n)$ \;
        }
        $j^* \gets \arg \max_{j \in \{1, \dots, k\}} e_j$ \tcp*{Box $k$ receives optimal fair cap $x_{j^*}$.}
        $t_k \gets x_{j^*} - \faircap{k}(0)$ \;
    }
\end{algorithm}
\begin{theorem}
    \label{thm:opt-algo-binary-boxes}
    Algorithm~\ref{alg:binary-boxes} computes an optimal exploration contract for binary boxes in polynomial time.
\end{theorem}
\begin{proof}
    For the proof, we show inductively that the decision of the algorithm in each iteration of the outer for-loop regarding box $k$ is optimal. Formally, there is an optimal ordering of all boxes that, when restricted to boxes $1,\ldots,k$, is identical to the ordering computed by the algorithm after iteration $k$.
    By Lemma~\ref{lem:binary-boxes:opt-order-implies-contract}, this implies that the selected fair cap for box $k$ is optimal (and, thus, the payment in $T$), as boxes are numbered and considered in a non-increasing order of their basic fair caps.

    For the rest of the proof we show the following statement. For every $k \in [n]$, let $\sigma_k$ denote the ordering of boxes after iteration $k$ of the outer for-loop in Algorithm~\ref{alg:binary-boxes}. There is an optimal ordering $\sigma^*$ such that the relative order of boxes $1, \dots, k$ is identical in $\sigma^*$ and $\sigma_k$.
 
    We proceed by induction. The statement is trivially true for $k = 1$. Now suppose the statement holds for every of the first $k-1$ iterations, and consider iteration $k$. Let $i^*$ denote the position of box $k$ w.r.t.\ boxes $1,\ldots,k$ in $\sigma^*$, and let $i$ denote its position $\sigma_k$. Note that position $i^*$ corresponds to a position $i_n^* \ge i^*$ in the overall ordering $\sigma^*$ (since for definition of $i^*$ we ignore boxes $k+1,\ldots,n$). Clearly, the statement holds when $i^* = i$.

    For the two remaining cases ($i^* > i$ and $i^* < i$) we show that the optimal ordering can be changed sequentially into one that has $i^* = i$ without decreasing the expected value for $\calP$. We discuss the proof for case 1: $i^* > i$. The proof of the other case is very similar (see Appendix~\ref{app:binary}). 
    
    Consider $i^* > i$. Now suppose box $k$ was brought to position $i$ (w.r.t. first $k$ boxes) in $\sigma^*$. The fair cap of $k$ must be strictly increased, since otherwise $k$ would receive the same payment and the same position in $\sigma_k$ and $\sigma^*$ (due to optimal tie-breaking by the agent in non-increasing order of principal value). Note that position $i$ corresponds to a position $i_n < i_n^*$ in $\sigma^*$.
        
    Consider the set $L$ of boxes located between positions $i_n$ and $i^*_n$ in $\sigma^*$. Consider the first box $r \in L$ with $r > k$ and the position $i_r$ in $\sigma^*$. Let $L_1$ be the boxes located before box $r$ in positions $i_n,\ldots,i_r-1$. We use $\smash{\tilde{p_1}}$ for the combined probability of selection in $L_1$ and $\smash{\tilde{b_1}}$ for the combined expected value of a selected box in $L_1$.
    
    Let $b_r^k$ be the remaining value of box $r$ after a payment that lifts the fair cap of $r$ to the basic fair cap of $k$. We split the analysis into two cases: $b_r^k \ge b_k$ and $b_r^k < b_k$. If $b_r^k \ge b_k$, consider the move of $r$ to position $i_n$. 
    Since there are only boxes $j < k$ in $L_1$, we could place box $k$ on position $i + |L_1| + 1$ in $\sigma_k$ -- right after boxes from $L_1$, similar to the position $i_r$ of box $r$ in $\sigma^*$.
    We use $\Delta_{i_r}$ to denote the payment required to lift the fair cap of box $k$ from its basic fair cap to the one required at position $i_r$. 
    Moreover, $\Delta_{i_n,i_r}$ is the additional payment required to lift the fair cap of box $k$ from the one at position $i_r$ to the one at $i_n$. 
    Since the algorithm places $k$ at $i_n$ instead of $i_r$, we know that
    \[
        p_k (b_k - \Delta_{i_r} - \Delta_{i_n, i_r})  + (1-p_k)\smash{\tilde{p_1}} \smash{\tilde{b_1}} \; \ge \; \smash{\tilde{p_1}}\smash{\tilde{b_1}} + (1-\smash{\tilde{p_1}}) p_k (b_k - \Delta_{i_r}),
    \]
    which implies 
    \[
        b_k \ge \smash{\tilde{b_1}} + \Delta_{i_r} + \Delta_{i_n,i_r}/\smash{\tilde{p_1}}.
    \]
    
    Suppose now we move $r$ from $i_r$ to $i_n$, then we have the same terms with $b_r^k$ instead of $b_k$. Clearly, since $b_r^k \ge b_k$ we know that the move is also profitable. This means that we can move $r$ to position $i_n$ without deteriorating the value of $\sigma^*$. Thereby we reduce the number of boxes $j \ge k$ between the positions of $k$ in $\sigma^*$ and $\sigma_k$.
    
    For the second case consider $b_r^k < b_k$. We here consider two moves in $\sigma^*$ -- either swap box $r$ right after box $k$, or swap box $k$ right before box $r$. Neither of these moves shall maintain the value of $\sigma^*$, and we show that this is impossible.
    
    Let $L_2$ be the boxes located between $r$ and $k$ in $\sigma^*$. We use $\smash{\tilde{p_2}}$ and $\smash{\tilde{b_2}}$ to denote selection probability and expected value of selected box, respectively. Now if we swap box $r$ after box $k$, a strict decrease in value yields
    \begin{align*}
        & p_r (b_r^k - \Delta_{i^*_n} - \Delta_{i^*_n,i_r}) + (1-p_r) \smash{\tilde{p_2}} \smash{\tilde{b_2}} + (1-p_r)(1-\smash{\tilde{p_2}}) p_k (b_k - \Delta_{i^*_n}) \\
        > \; & \smash{\tilde{p_2}} \smash{\tilde{b_2}} + (1-\smash{\tilde{p_2}}) p_k (b_k - \Delta_{i^*_n}) + (1-\smash{\tilde{p_2}})(1-p_k) p_r (b_r^k - \Delta_{i^*_n}).
    \end{align*}
    
    Note that ensuring the same fair cap as box $k$ at position $i^*_n$ is enough, since $b_r^k < b_k$ and box $r$ is then sorted after box $k$. We obtain
    \begin{equation*}
        \smash{\tilde{p_2}} b_r^k + p_k(1-\smash{\tilde{p_2}}) b_r^k \; > \; \smash{\tilde{p_2}} \smash{\tilde{b_2}} + p_k (1-\smash{\tilde{p_2}}) b_k + \smash{\tilde{p_2}} \Delta_{i^*_n} + \Delta_{i^*_n,i_r}
    \end{equation*}
    and, since $b_r^k < b_k$,
    \begin{equation}
        \label{eq:1}
        b_r^k > \smash{\tilde{b_2}} + \Delta_{i^*_n} + \Delta_{i^*_n,i_r}/\smash{\tilde{p_2}}.
    \end{equation}
    
    For the other move, we see that a strict decrease in value yields
    \begin{align*}
        & p_r (b_r^k - \Delta_{i^*_n} - \Delta_{i^*_n,i_r}) + (1-p_r) \smash{\tilde{p_2}} \smash{\tilde{b_2}} + (1-p_r)(1-\smash{\tilde{p_2}}) p_k (b_k - \Delta_{i^*_n}) \\
        > \; & p_k (b_k - \Delta_{i^*_n} - \Delta_{i^*_n,i_r}) + (1-p_k) p_r (b_r^k - \Delta_{i^*_n} - \Delta_{i^*_n,i_r}) + (1-p_k)(1-p_r)\smash{\tilde{p_2}} \smash{\tilde{b_2}}.
    \end{align*}
    
    Note that ensuring the same fair cap as box $r$ at position $i_r$ is enough, since $b_r^k < b_k$ and box $k$ is then sorted before box $r$. We obtain
    \begin{equation*}
        (- p_r - \smash{\tilde{p_2}} + p_r \smash{\tilde{p_2}}) b_k + (1-p_r) \smash{\tilde{p_2}} \Delta_{i^*_n} \; > \; -(1 - p_r) \Delta_{i_n^*,i_r} - p_r b_r^k - \smash{\tilde{p_2}} \smash{\tilde{b_2}} +p_r \smash{\tilde{p_2}} \smash{\tilde{b_2}},
    \end{equation*}
    so
    \begin{equation*}
        (1-p_r) \smash{\tilde{p_2}} \smash{\tilde{b_2}} + (1-p_r) \smash{\tilde{p_2}} \Delta_{i^*_n} + (1 - p_r) \Delta_{i_n^*,i_r} \; > \; \smash{\tilde{p_2}} (1- p_r) b_k + p_r(b_k - b_r^k).
    \end{equation*}
    This implies $p_r < 1$, since otherwise we derive $b_k - b_r^k < 0$, a contradiction. Now, if $p_r \in (0,1)$, $\smash{\tilde{p_2}} \in (0,1]$, and $b_r^k < b_k$, the above implies
    \begin{equation}
        \smash{\tilde{b_2}} + \Delta_{i^*_n} + \Delta_2/\smash{\tilde{p_2}} \; > \; b_k. \label{eq:2}
    \end{equation}
    Equations~\eqref{eq:1} and~\eqref{eq:2} imply a contradiction with $b_r^k < b_k$.
    
    At least one swap (either moving $r$ right after $k$, or moving $k$ right before $r$) must be non-deteriorating in terms of solution value. Using this swap, we can change $\sigma^*$ and decrease the set of agents $j \ge k$ between $i^*_n$ and $i_n$ without loss in value.
    
    Now using these swaps (either the one for $b_r^k \ge b_k$ or one of the two for $b_r^k < b_k$) iteratively, we can turn $\sigma^*$ into an ordering such that there are only boxes $j < k$ in $L$. Then moving box $k$ from $i^*_n$ to $i_n$ is not harmful. We let $p_\ell$ denote the overall probability that a box in $L$ is chosen and $b_{\ell}$ the expected value. Let $\Delta_{i,i^*}$ be the additional payment required to raise the fair cap from the one of position $i^*$ to the one of position $i$.
    
    Indeed, our algorithm could have placed box $k$ at position $i^*$, but preferred to place $k$ in position $i < i^*$. This means that
    \[
        p_k(b_k - \Delta_{i,i^*}) + (1-p_k) p_{\ell} b_{\ell} \ge p_{\ell}b_{\ell} + (1-p_{\ell})p_k b_k,
    \]
    which also reflects the change in value if the move is executed in $\sigma^*$. This proves the inductive step when $i^* > i$.
\end{proof}

\subsection{I.I.D.\ with Single Positive Prize for Principal}
\label{sec:iid}

Finally, we study the subclass of the problem where all distributions are identical, with the further restriction that there is only a single positive prize for the principal. Even this restriction has a perhaps surprisingly complicated solution. We show the following result.

\begin{theorem}\label{thm:iid}
    There exists an algorithm that computes an optimal exploration contract for the i.i.d.\ case with a single positive prize for the principal in polynomial time.
\end{theorem}

We introduce some simplified notation for this subclass. All of the $n$ boxes have an identical distribution over prizes $\{0\}\cup[m]=\{0,1,\dots,m\}$ and the same opening cost $c \ge 0$.
For each prize $j \in \{0,1,\dots,m\}$, the utilities for \calA and \calP are $a_j$ and $b_j$, respectively.
Without loss of generality, prize $0$ is the only prize for which \calP has a positive value, i.e., $b_j = 0$ for all $j \ge 1$.
The probability that prize $j$ is drawn when a box is opened is denoted by $q_j \in [0,1]$.
For the ease of notation, we define $p := q_0$ and $v := b_0$. For each box $i \in [n]$, $\mathcal P$ defines a (non-negative) payment of $t_i \le v$ for prize $0$.
We may assume $v > 0$ and $p > 0$ as otherwise $t_1=\dots=t_n=0$ would be the unique transfers.

We make two further assumptions that are without loss of generality. 
First, the values $a_j$ for $j \ge 1$ are unique: If there are $j_1, j_2 \in [m]$ with $a_{j_1} = a_{j_2}$, replace both prizes by a new one with probability $q_{j_1} + q_{j_2}$. Second, we assume $a_1 \le a_2 \le \dots \le a_m$.

Similarly as in Section~\ref{sec:binary}, let $\faircap{}(t)$ be the fair cap of a box for $\calA$ when payment $t \geq 0$ is made for prize $0$ of that box. We call $\faircap{}(0)$ the \emph{basic fair cap}. Note that there is an optimal contract in which each fair cap is either $\faircap{}(0)$ or some value $a_j$, where $j\in[m]$, with $a_j>\faircap{}(0)$. The reason is that, if that were not the case, we could decrease the payment for each box, until that property is satisfied while keeping the policy optimal.
Importantly, the agent value of outcome 0 (which is given by $a_0 + t_i$) need not be considered here:
Whenever outcome 0 is drawn from a box $i$ with $\faircap{i} > \faircap{}(0)$, the process stops immediately since, by definition of the fair cap and the index policy, it must hold $a_0 + t_i > \faircap{i} \geq \faircap{i+1}$, where $\faircap{i+1}$ is the fair cap of the next box.

Observe that $\faircap{}(0)$ can be achieved with different payments $t$, as long as $a_0 + t \le \faircap{}(0)$. This may influence the stopping behavior of \calA as the payment may determine the order of $a_0+t$ and $a_j$ for some prize $j$. 
To achieve some fair cap larger than $\faircap{}(0)$, however, there is a unique payment that achieves this fair cap because changing a values above the fair cap that occurs with non-zero probability always changes the fair cap by definition.

Now consider an optimal contract. Index the boxes in the order in which they are opened, and let $\faircap{1} \ge \faircap{2} \ge \dots \ge \faircap{n}$ be the corresponding fair caps. We first concentrate on a potential set of \emph{basic} boxes with $\faircap{}(0)$ in the final rounds of the process.

If $a_0 \ge \faircap{}(0)$, then every positive payment of $\mathcal P$ would increase the fair cap to above $\faircap{}(0)$. Therefore, if basic boxes with fair cap $\faircap{}(0)$ exist in an optimal contract, $t_i = 0$ must hold for all these boxes.

Otherwise, if $a_0 < \faircap{}(0)$, suppose $j^* \in [m]$ is the best other prize for $\mathcal A$ that is not exceeding the basic fair cap, i.e., $j^* := \arg \max_{j \in [m]} \{a_j : a_j \le \faircap{}(0)\}$. Observe that for every basic box $i$, an optimal payment $t_i$ either fulfills $a_0 + t_i = a_j$ for some $j \in [j^*]$, or $a_0 + t_i = \faircap{}(0)$. W.l.o.g., when prize 0 is drawn in a basic box $i$ with $a_0 + t_i = \faircap{}(0)$, the process stops (by tie-breaking in favor of $\mathcal P$, who cannot improve), and when an outcome $j \ge 1$ with $a_j = \faircap{}(0)$ is drawn in any round before $n$, the process continues (again by tie-breaking in favor of $\mathcal P$, who gets $0$ if $a_j$ is accepted).

We next show that all basic boxes $i$ with $a_0 + t_i = \faircap{}(0)$ can be assumed to be opened first among the set of basic boxes. The argument is an exchange argument and given in Appendix~\ref{app:iid}.

\begin{lemma}\label{lem:iid-p2}
	There is an optimal contract where all basic boxes $\ell$ with $a_0 + t_{\ell} = \faircap{}(0)$ are opened first among the set of basic boxes.
\end{lemma}

We now divide the process into two phases. In every round $i$ of Phase 1, $\faircap{i} \ge \faircap{}(0)$ and prize 0 gets accepted by $\mathcal{A}$ immediately if it is found. In Phase 2, $\faircap{i} = \faircap{}(0)$ and prize 0 is not accepted immediately (only if in the last round a prize 0 represents the best one for $\mathcal{A}$). Suppose Phase 1 contains boxes $1, \dots, k$ and Phase 2 boxes $k+1, \dots, n$. Note that either phase might be empty.

Before analyzing the process, we show an example that shows the perhaps counter-intuitive fact that the optimal contract might indeed involve two such phases, each of substantial length. Moreover, even though all boxes have the same distribution and the same fair cap, in the optimal contract they are assigned one of two different payments depending on their position in the inspection sequence.

\begin{example} \rm
Consider an instance with $m = 2$. Prize 0 has value $a_0 = 0$ and $b_0 = 1+\alpha$ for some $\alpha > 0$ determined below, as well as $q_0 = 1/n$. The other prizes have $a_1 = 2$, $q_1 = 1/n$ and $a_2 = 0$, $q_2 = 1-2/n$. The inspection cost is $c = 1/n$, and the basic fair cap is $\faircap{}(0) = 1$.

Clearly, if prize 1 is found in a box, then $\mathcal{A}$ always accepts. If prize 2 is found, $\mathcal{A}$ continues to inspect further boxes. Suppose $\mathcal A$ draws a prize 0 in the first round. It is rather unlikely that $\mathcal{A}$ will run through all remaining $n-1$ boxes without finding a prize 1 and eventually accept prize 0 from round 1. Specifically, this probability is $(1-1/n)^{n-1}$. Now $\mathcal P$ can set $t_1 = 1 = \faircap{}(0)$, leading to direct acceptance by $\mathcal A$ and immediate profit of $\alpha$. Alternatively, $\mathcal P$ can decide to set $t_1 = 0$ (and, in an optimal contract, then $t_i = 0$ throughout) gambling for a high profit of $1+\alpha$ in the end. The gambling strategy is unprofitable if 
\[
    \alpha \ge (1+\alpha) \cdot (1-1/n)^{n-1} \quad \text{ or, equivalently, } \quad \alpha \ge \frac{(1-1/n)^{n-1}}{1-(1-1/n)^{n-1}}\enspace.
\]
Suppose $\alpha$ satisfies this inequality, then round 1 is part of Phase 1 in an optimal contract. Via the same calculation, we can determine further rounds $i$ of Phase 1. $\mathcal A$ accepts prizes 0 and 1 in all earlier rounds $1,\ldots,i-1$, so reaching round $i$ is only possible if prize 2 has been found in all these rounds. 

Now suppose for some $k \in \{1,\ldots,n-1\}$
\[
\alpha = \frac{(1-1/n)^{k}}{1-(1-1/n)^{k}}\enspace.
\]
Then $\mathcal P$ wants to set $t_i = 1$ for all $i\le k$, i.e., this becomes Phase 1 with direct acceptance of prize 0. For all rounds $i \ge k+1$, it is more profitable to gamble that no prize 1 will arrive in the remaining rounds and $\mathcal P$ will secure a value of $1+\alpha$. This represents Phase 2. \hfill $\blacksquare$
\end{example}

Turning to the analysis, we start by examining Phase 1. Using a relatively straightforward exchange argument given in Appendix~\ref{app:iid}, we show the following.

\begin{lemma}\label{lem:iid-p1}
	There is an optimal contract in which all boxes of Phase 1 have the same fair cap.
	\label{lem:iid-single-P-option:identical-faircaps-phase1}
\end{lemma}

In Phase 2, all boxes have fair cap $\faircap{}(0)$, but this does not exclude the possibility of different payments. If the first box of Phase 2, box $k+1$, is about to be opened, the highest agent value observed so far (if any) is at most $\faircap{}(0)$. In particular, prize 0 was never observed at that point, since this would have led to acceptance of $\mathcal{A}$ in an earlier round.

During Phase 2, the process only stops early when an outcome $j > j^*$ is drawn. In this case, the utility for $\mathcal P$ is 0. If the process does not stop early, $\mathcal P$ receives a utility only if outcome 0 is the best outcome for the agent among all $n$ outcomes.

$\mathcal P$ could have different payments for different boxes as long as, for every box $i$, there is some $j \in [j^*]$ such that $a_0 + t_i = a_j$.
If an outcome from a box with payment $t = a_j - a_0 \le \faircap{}(0) - a_0$ (for some $j \in [j^*]$) is the maximum after $n$ rounds, then
\begin{itemize}
	\item outcome 0 must have been drawn from any number of boxes that have this payment $t$, and
	\item in all boxes with a payment of at most $t$, some outcome $0 \le j' < j$ was drawn, and
	\item in all remaining boxes, some outcome $1 \le j' < j$ was drawn.
\end{itemize}
Let $n_j := |\{i \in [n] : a_0 + t_i = a_j\}|$ denote the number of boxes of Phase 2 with such payment, for every $j \in [j^*]$.
Furthermore, for every $j \in [j^*]$, let $Q_j := \sum_{j' = 1}^j q_{j'}$ denote the combined probability of all outcomes (other than 0) that are not better for $\mathcal A$, and $N_j = \sum_{j' = 1}^j n_{j'}$ denote the number of boxes with the according payment.
The probability that outcome 0 with a payment of $t = a_j - a_0$ is the best outcome after $n$ rounds is then given by
\begin{align*}
		& \sum_{i = 1}^{n_j}\binom {n_j} i \cdot p^i \cdot (p + Q_j)^{N_{j-1}} \cdot Q_j^{n - N_{j-1} - i}\\
	&= \sum_{i = 1}^{n_j}\binom {n_j} i \cdot p^i \cdot \left(\frac {p + Q_j} {Q_j}\right)^{N_{j-1}} \cdot Q_j^{n + n_j - n_j - i}\\
	&= \sum_{i = 1}^{n_j}\binom {n_j} i \cdot p^i \cdot \left(\frac {p} {Q_j} + 1\right)^{N_{j-1}} \cdot Q_j^{n - n_j} \cdot Q_j^{n_j - i}\\
	&= \left(\frac {p} {Q_j} + 1\right)^{N_{j-1}} \cdot  Q_j^{n - n_j} \cdot \left(\sum_{i = 0}^{n_j}\binom {n_j} i \cdot p^i \cdot Q_j^{n_j - i} - Q^{n_j}\right)\\
	&= \left(\frac {p} {Q_j} + 1\right)^{N_{j-1}} \cdot  Q_j^{n - n_j} \cdot \left( (p + Q_j)^{n_j} - Q^{n_j}\right)\\
	&= \left(\frac {p} {Q_j} + 1\right)^{N_{j-1}} \cdot  Q_j^{n} \cdot \left( \left(\frac{p} {Q_j} + 1\right)^{n_j} - 1\right)\\
	&= Q_j^{n} \cdot \left( \left(\frac{p} {Q_j} + 1\right)^{N_j} - \left(\frac{p} {Q_j} + 1\right)^{N_{j-1}}\right).
\end{align*}
Thus, the expected utility for $\mathcal P$ for Phase 2 is 
\[\sum_{j = 1}^{j^*} v_j \cdot  Q_j^{n} \cdot \left( \left(\frac{p} {Q_j} + 1\right)^{N_j} - \left(\frac{p} {Q_j} + 1\right)^{N_{j-1}}\right).
\]

With this at hand, we show the following lemma in Appendix~\ref{app:iid}.

\begin{lemma}
    \label{lem:iid-final}
	There is an optimal contract in which all boxes of Phase 2 have the same payment for prize 0. There is $j^+ \in [j^*]$ such that $n_{j^+} = n-k$ (and $n_j = 0$ for all $j \neq j^+$).
\end{lemma}

By the previous lemmas, it suffices to enumerate all combinations of the length of Phase 1, the fair cap in Phase 1, and the unique payment of prize 0 for Phase 2 (determined by a prize $j$ to be targeted, a payment of 0, or a payment resulting in value $\faircap{}(0)$ for $\mathcal A$ for prize 0). These are polynomially many combinations, and for each combination the expected utility of \calP for the resulting contract can be computed in polynomial time. Theorem~\ref{thm:iid} follows.

\bibliography{arxiv_V3}

\clearpage
\appendix 

\section{Missing Material from Section~\ref{sec:prelims}}
\label{app:tieBreak}

\begin{proof}[Proof of Theorem~\ref{thm:best-order}.]
We first argue that the ambiguity regarding the choice of the fair caps can be resolved as claimed. This only concerns the boxes $i$ with $c_i=0$. Note that $\faircap{i}$ is consistent with the definition if and only if $\faircap{i} \ge \max_{j:p_{ij}>0} \{a_{ij} + t_{ij}\}$. We argue that it is optimal for \calP to choose these fair caps large enough for all the corresponding boxes to be opened at the very beginning, regardless of the found prizes. Note that it would only be consistent with the index policy \emph{not} to open a box $i$ with $c_i=0$ if the prize observed so far yielded a utility of at least $\max_{j:p_{ij}>0} \{a_{ij} + t_{ij}\}$ for \calA. Since the latter quantity is precisely the largest utility that a prize in box $i$ could yield for \calA, \calA would only select a prize from box $i$ when it is favorable for \calP. Therefore, it is never detrimental for \calP if box $i$ is opened, so box $i$ may as well be opened at the very beginning.

Thus, in the remainder of this proof, we consider the boxes with strictly positive cost. 
We partition these boxes into \emph{phases} according to their fair caps, i.e., all boxes with an identical fair cap belong to the same phase. The optimal order of phases is unambiguously determined by the index policy. Hence with respect to optimal tie-breaking, considering the \calP-optimal action within each phase separately, conditioned on the previous realizations, suffices.

Let $\faircap{}$ be the fair cap of any phase.
If we consider the first phase, we define $a^* = 0$.
Otherwise, let option $j'$ from box $i'$ be a best option (for \calA) from previous phases, with ties broken in favor of \calP.
Then $a^* := a_{i'j'}+t_{i'j'}$ is the utility for \calA for this option.
Note that the current phase is only reached when $a^* \le \faircap{}$, and does not need to be considered otherwise.

For each box $i$ of the current phase, we partition the set of options with positive probability depending on the behavior of \calA during the phase, which is given by \calA's index policy.
Let \[\Omega_i^{>} := \{j \in [m] ~:~ p_{ij} > 0,~a_{ij} + t_{ij} > \faircap{i}\}\] denote the set of options from box $i$ where \calA would stop immediately.
Furthermore, let \[\Omega_i^{<} := \{j \in [m] ~:~ ~p_{ij} > 0,~a_{ij} + t_{ij} < \faircap{i}\}\] denote the set of options from box $i$ that \calA would never take immediately (unless all boxes of the current phase are opened). 
Upon observing all other options, \calA is indifferent, and they are collected in the set
\[\Omega_i^{=} := \{j \in [m] ~:~ p_{ij} > 0,~a_{ij} + t_{ij} = \faircap{i}\}.\]
Note that if $a^* = \faircap{}$, \calA can stop at any time and take the previously revealed value of $a^*$.
Additionally, note that the set $\Omega_i^{>}$ is never empty, as boxes have strictly positive opening costs.

We reformulate the problem of finding a \calP-optimal policy among all \calA-optimal policies for the current phase as follows:
There are $n$ boxes in the current phase, and there is a \emph{fallback option} with value $v^*$, where
\[v^* = \begin{cases}
    0,& \text{ if } a^* < \faircap{},\\
    b_{i'j'} - t_{i'j'},& \text{ otherwise.}
\end{cases}\]
For each box $i \in [n]$, there are $m_i := |\Omega_i^{=}| + 2$ options that are drawn independently from a known distribution.
Option 1 has probability $\tilde p_{i,1}$ and value $v_{i,1}$, and combines all options from $\Omega_i^{>}$ that make \calA stop immediately.
The value is the conditioned expected utility of those options for \calP. 
More formally, 
\[\tilde p_{i,1} := \sum_{j \in \Omega_i^{>}} p_{ij} \quad \text{and} \quad v_{i,1} := \frac {\sum_{j \in \Omega_i^{>}} (b_{ij} - t_{ij}) \cdot p_{ij}} {\sum_{j \in \Omega_i^{>}} p_{ij}}.\]
Furthermore, the options $2, \dots, m_i - 1$ are given by the set $\Omega_i^{=}$. 
Their probabilities $\tilde p_{i,j}$ are identical to to their respective counterpart in  $\Omega_i^{=}$, while their values $v_{i,j}$ are given by the respective utility for \calP.
We assume w.l.o.g. that $v_{i,2} \ge v_{i,3} \ge \dots \ge v_{i,m_i-1} \ge 0$.
Finally, the last option $m_i$ combines all options in $\Omega_i^{<}$ that \calA will never take before having inspected all $n$ boxes of the phase.
Let \[\tilde p_{i,m_i} := \sum_{j \in \Omega_i^{<}} p_{ij} \quad \text{and} \quad v_{i,m_i} := 0.\]
The boxes can be opened in any order without opening costs.
Whenever outcome 1 is drawn from some box $i$, the process stops immediately with value $v_{i,1}$.
Otherwise, the process can be stopped voluntarily in case there is a strictly positive value available in an opened box or as fallback value $v^*$, and it stops after $n$ rounds at the latest.
When stopping without having option 1 drawn, we receive the maximum value from all opened boxes (including the fallback value $v^*$).

We need to argue why our formulated problem is equivalent to the problem of finding a \calP-optimal tie-breaking rule.
When all boxes of the current phase are opened, the expected utilities are not always preserved:
In the original process, options from the set $\Omega_i^{<}$ of some box $i$ can possibly yield a positive value for \calP after the phase is completed, but it holds $v_{i,m_i} = 0$ in the reformulated process.
Crucially, the values within $\Omega_i^{<}$ never affect tie-breaking: 
In order to obtain a positive value with an option from $\Omega_i^{<}$, the following must hold:
\begin{enumerate}  
    \item \emph{Every} box from the current phase has to be opened.
    \item \emph{Every} box $i'$ from the current phase has to contain an option from $\Omega_{i'}^{<}$, as any option from $\Omega_{i'}^{=}$ or $\Omega_{i'}^{>}$ would be preferred by \calA.
    \item $v^* = 0$, because a fallback option with positive value for \calP would be preferred by \calA.
\end{enumerate}
As all boxes have to be opened, the opening order is irrelevant.
All decisions regarding the tie-breaking are made without incorporating the actual value of drawn options from $\Omega_{i'}^{<}$ of any box $i'$: 
Whenever \calA is indifferent, it must be upon observing an option from $\Omega_{i'}^{=}$ after opening some box $i'$, or due to a positive fallback value. 
But then, an option from $\Omega_i^{<}$ will never be considered by \calA.

Therefore, the choices for \calP-optimal tie-breaking are preserved with our formulated problem.
The same is true for the original expected utilities whenever they can influence the optimal action (opening order and stopping behavior).
Therefore an optimal action for the formulated problem constitutes an optimal tie-breaking rule.

We apply further adjustments to derive an optimal policy for our formulated problem.
For every box $i \in [n]$, let 
\[
    \tau_i := \max_{k \in [m_i-1]}\frac {\sum_{j=1}^k \tilde p_{ij} v_{ij}} {\sum_{j = 1}^k \tilde p_{ij}},
\]
i.e., $\tau_i$ is the maximum expected value for $\calP$ for an option from box $i$ conditioned on the event that an option from that box is taken (subject to being consistent with \calA's index policy).
Consequently, $\tau_i \ge v_{i,1}$.

An equivalent formulation of the problem is the following: 
(Weakly) increase the value of option $1$ to $v_{i,1}' := \tau_i$, but, whenever option 1 is drawn, costs of $c_{i,1} := \tau_i - v_{i,1}\geq 0$ are incurred. 
All other options keep their value and no costs occur for them, i.e., for $j > 1$, set $v_{i,j}' := v_{i,j}$ and $c_{i,j} := 0$. 

Similarly as in~\cite{KleinbergWW16}, to analyze the (expected) performance of any policy for this problem, let
\begin{align*}
    J(i) &:= \text{(random) index of the option that is in box } i \in [n],\\
    I_i &:= \mathbf 1_{\text{Box } i \in [n] \text{ is opened}},\\
    A_i &:= \mathbf 1_{\text{Option } J(i) \text{ from box } i \in [n] \text{ is taken}}.
\end{align*}
We can only select option $J(i)$ from box $i$ if the box was opened before:
\begin{align}
    A_i = 1 \implies I_i = 1 \quad \text{for all}~i \in [n].
    \label{eqn:open-before-accept}
\end{align}
Furthermore, if box $i \in [n]$ is opened and contains option 1, then we have to stop and select it:
\begin{align}
    I_i = 1 \text{ and } J(i) = 1 \implies A_i = 1 \quad \text{for all}~i \in [n].
    \label{eqn:always-take-outcome-1}
\end{align}
The utility $u_\calP$ that the policy yields for $\calP$ can be written as
\[
    u_\calP = \sum_{i = 1}^n A_i \cdot v_{i, J(i)} = \sum_{i = 1}^n A_i \cdot \left( v'_{i, J(i)} - c_{i,J(i)} \right).
\]
By linearity of expectation, this means:

\begin{align}
    \mathbb E[u_\calP] 
    &= \mathbb{E}\left[\sum_{i = 1}^n A_i \cdot \left( v'_{i,J(i)} - c_{i,J(i)} \right)\right] \notag\\
    &= \mathbb{E}\left[\sum_{i = 1}^n A_i \cdot v'_{i,J(i)}\right] - \mathbb{E}\left[\sum_{i = 1}^n A_i \cdot c_{i,J(i)}\right]\notag\\
  %  &= \mathbb{E}\left[\sum_{i = 1}^n A_i \cdot v'_{i,J(i)}\right] - \sum_{i = 1}^n \sum_{(j,k) \in [m_i] \times \{0,1\}} \Pr[J(i) = j \cap A_i = k] \cdot k \cdot c_{i,j}\notag\\
    &= \mathbb{E}\left[\sum_{i = 1}^n A_i \cdot v'_{i,J(i)}\right] - \sum_{i = 1}^n \Pr[J(i) = 1 \cap A_i = 1] \cdot c_{i,1}\notag\\
    &= \mathbb{E}\left[\sum_{i = 1}^n A_i \cdot v'_{i,J(i)}\right] - \sum_{i = 1}^n \Pr[J(i) = 1 \cap I_i = 1] \cdot c_{i,1}\notag\\
    &= \mathbb{E}\left[\sum_{i = 1}^n A_i \cdot v'_{i,J(i)}\right] - \sum_{i = 1}^n \tilde p_{i,1} \cdot \Pr[I_i = 1] \cdot c_{i,1}\notag\\
    &= \mathbb{E}\left[\sum_{i = 1}^n A_i \cdot v'_{i,J(i)}\right] - \mathbb{E}\left[\sum_{i = 1}^n I_i \cdot \tilde p_{i,1} \cdot c_{i,1}\right].\label{eq:ambiguity-mod-pandora}
\end{align}
We use that $c_{i,J(i)} = 0$ unless $J(i) = 1$, but in this case, (\ref{eqn:open-before-accept}) and (\ref{eqn:always-take-outcome-1}) imply $A_i = I_i$.
Then we use that $I_i = 1$ and $J(i) = 1$ are independent events.

We now create an instance for the (original) Pandora's Box problem.
We keep values $v'_{i,1} \dots, v'_{i,m_i}$ and probabilities $\tilde p_{i,1}, \dots, \tilde p_{i,m_i}$ for every box $i \in [n]$. 
Moreover, we have the opening costs of $c_i := \tilde p_{i,1} \cdot c_{i,1}$ for every box $i \in [n]$. 
Recall that, in Pandora's Box problem there are no costs for certain options or obligations to accept them.
In particular, if we define $\tilde A_i$ and $\tilde I_i$ as the analogs of the above indicator variables for a policy for the Pandora's Box problem, then the analog of~\eqref{eqn:always-take-outcome-1} is \emph{not} required. The utility that the policy obtains in expectation is on the other hand (c.f.~\cite{KleinbergWW16})
\begin{equation}\label{eq:ambiguity-pandora}
    \mathbb E \left[\sum_{i = 1}^n \tilde A_i \cdot v'_{i, J(i)} - \tilde I_i \cdot c_i\right]
    =  \mathbb E \left[ \sum_{i = 1}^n \tilde A_i \cdot v'_{i, J(i)} \right] - \mathbb E \left[ \sum_{i = 1}^n \tilde I_i \cdot \tilde p_{i,1} \cdot c_{i,1} \right],
\end{equation}
which is the analog of~\eqref{eq:ambiguity-mod-pandora}. As such, the Pandora's Box instance is a relaxation of the former instance.
There does, however, exist an optimal policy for the Pandora's Box instance such that $\tilde A_i$ and $\tilde I_i$ fulfill the analog of~\eqref{eqn:always-take-outcome-1}: The fair cap $\faircap{i}$ of box $i$ is $\tau_i$ for every $i \in [n]$. 

This condition can be seen as follows. Recall that by definition, 
\[\tau_i = \left(\sum_{j = 1}^{k^*} \tilde p_{ij} v_{ij} \right) / \left({\sum_{j = 1}^{k^*} \tilde p_{ij}} \right)\] 
for some $k^* \in [m_i - 1]$, such that $v'_{i,j} \ge \tau_i$ for all $j \in [k^*]$, and $v'_{i,j} \le \tau_i$ for all $j \in [m_i] \setminus [k^*]$. Otherwise, $\tau_i$ was not chosen as the correct maximum. 
If $k^* = 1$, it follows $\tau_i = v_{i1}$ and therefore $c_i = 0$ in our constructed instance. In this case, the fair cap is ambiguous and must only fulfill $\faircap{i} \ge \max_{j \in [m_i]} v_{ij} = v_{i1}$, which is true.
Otherwise, it holds $\tau_i > v_{i1}$ and therefore $c_i > 0$, and the fair cap is unique.
We now prove that $\faircap{i} = \tau_i$ fulfills the definition of the fair cap:
\[
    \sum_{j = 1}^{m_i} \tilde p_{ij} \cdot \max\{0, v'_{ij} - \tau_{i}\} = c_i 
    ~\Leftrightarrow~
    \sum_{j = 1}^{k^*} \tilde p_{ij} \cdot (v'_{ij} - \tau_{i}) = c_i ~\Leftrightarrow~ \frac {\sum_{j = 1}^{k^*} \tilde p_{ij} v'_{ij} - c_i} {\sum_{j = 1}^{k^*} \tilde p_{ij}} = \tau_i.
\]
This equals the definition of $\tau_i$, because in the numerator, we have
\[
    \sum_{j = 1}^{k^*} \tilde p_{ij} v'_{ij} - c_i 
    = \sum_{j = 2}^{k^*} \tilde p_{ij} v'_{ij} + \tilde p_{i1} v'_{i1} - c_i
    = \sum_{j = 2}^{k^*} \tilde p_{ij} v_{ij} + \tilde p_{i1} \tau_{i} - \tilde p_{i1} \cdot (\tau_i - v_{i1})
    = \sum_{j = 1}^{k^*} \tilde p_{ij} v_{ij}.
\]

Thus $v'_{i,1} = \faircap{i}$ for all $i \in [n]$, and stopping with option $1$ when it is realized is always in accordance with the index policy.
Therefore, the following policy is optimal for the former problem (with the obligation to stop whenever option 1 is observed): Consider boxes in the order of $\tau_i$. When box $i$ is considered and option 1 was not drawn before, open box $i$ unless the maximum value from an opened box or the fallback value is greater than $\tau_i$. Always stop when option 1 is drawn. Break any remaining ties arbitrarily.
\end{proof}

\section{Missing Material from Section~\ref{sec:binary}}
\label{app:binary}

\begin{proof}[Proof of Theorem~\ref{thm:opt-algo-binary-boxes} (continued)]

    Now consider $i^* < i$. Suppose box $k$ was brought to a position $i_n$ in $\sigma^*$ such that it is on position $i$ w.r.t. the first $k$ boxes in $\sigma^*$. Then the fair cap of box $k$ must be strictly decreased, since otherwise it would receive the same payment and the same (relative) position in $\sigma_k$ and $\sigma^*$ (among the first $k$ boxes), see Lemma~\ref{lem:binary-boxes:opt-order-implies-contract}.

    Note that the fair cap necessary for position $i_n$ is still at least the basic fair cap of box $k$. Otherwise box $k$ was not the next box considered by the algorithm).

    Let $L$ denote the set of boxes between $i_n^*$ (current position of box $k$ in $\sigma^*$) and $i_n$ in $\sigma^*$, and let $r \in L$ denote the \emph{last} box in $\sigma^*$ with an index greater than $k$. The position of $r$ in $\sigma^*$ is denoted by $i_r$.

    Let $L_1 \subseteq L$ denote the set of subsequent boxes, i.e., boxes after $i_r$ and before $i_n$. Note that all boxes in $L_1$ are also present in $\sigma_k$. Let $\smash{\tilde{p_1}}$ denote the combined probability for a positive prize in any $L_1$-box and let $\smash{\tilde{b_1}}$ denote the combined expected value (for \calP) for a positive prize from $L_1$. Furthermore, let $b_r^k$ be the remaining value (for \calP) of the positive prize of box $r$ with a payment that raises its fair cap to the basic fair cap of box $k$.

    If $b_r^k \le b_k$, we can move $r$ from $i_r$ to $i_n$ (i.e. behind $L_1$) without decreasing the expected utility of $\sigma^*$. To prove this, we consider a different swap in $\sigma_k$, namely, the movement of box $k$ from its current position $i$ to the position right before $L_1$. We know that this swap cannot be strictly profitable, as position $i$ is an optimal position for box $k$ in $\sigma_k$.

    Let $\Delta_{i}$ denote the payment that lifts box $k$ from its basic fair cap to position $i$, and let $\Delta_{L_1, i}$ denote the additional payment to lift box $k$ from position $i$ right to the front of $L_1$. Then by the former argument, it holds
    \[
        \smash{\tilde{p_1}}\smash{\tilde{b_1}} + (1-\smash{\tilde{p_1}})p_k(b_k - \Delta_i) \ge p_k(b_k - \Delta_i - \Delta_{L_1,i}) + (1-p_k)\smash{\tilde{p_1}}\smash{\tilde{b_1}}.
    \]
    This is equivalent to
    \begin{equation}
        \smash{\tilde{b_1}} + \Delta_i + \frac {\Delta_{L_1, i}} {\smash{\tilde{p_1}}} \ge b_k
        \label{Eq:ForwardSwapOfBoxKNotProfitable}
    \end{equation}
    If we would swap box $r$ in $\sigma^*$ between positions $i_n$ and $i_r$, we would yield the same equation, with $b_r^k$ instead of $b_k$. This holds because we have the same fair caps for all boxes $\{1, \dots, k-1\}$ (by Lemma~\ref{lem:binary-boxes:opt-order-implies-contract}), so the same payments are required within $\sigma^*$. However, as $b_k \ge b_r^k$ and Equation~\eqref{Eq:ForwardSwapOfBoxKNotProfitable} was true already for $b_k$, it must be also true for $b_r^k$. Hence swapping box $r$ from $i_r$ to $i_n$ cannot decrease the utility.

    If $b_r^k > b_k$, we need to consider two swaps for box $r$ in $\sigma^*$ that cannot be both decreasing the utility for the principal. More precisely, we can either swap box $r$ right before box $k$ or we can swap box $k$ right after box $r$ without deterioration. Let $L_2 \subseteq L$ denote the boxes between $i_n^*$ (where box $k$ is) and $i_r$ (where box $r$ is) in $\sigma^*$. Again, combined probability for a positive prize within $L_2$ is $\smash{\tilde{p_2}} > 0$ and combined expected value is $\smash{\tilde{b_2}}$. Let $\Delta_{i_r}$ denote the payment that lifts box $k$ to position $i_r$, and let $\Delta_{i_r, i_n^*}$ denote the additional payment lifting box $k$ from $i_r$ to its current position in $\sigma^*$, namely, $i_n^*$.

    Suppose swapping box $r$ right before box $k$ decreases the utility, then it holds
    \begin{align*}
         &~ p_k(b_k - \Delta_{i_r} - \Delta_{i_r, i_n^*}) + (1-p_k)\smash{\tilde{p_2}}\smash{\tilde{b_2}} + (1-p_k)(1-\smash{\tilde{p_2}})p_r(b_r^k - \Delta_{i_r})\\
        >&~ p_r(b_r^k - \Delta_{i_r} - \Delta_{i_r, i_n^*}) + (1-p_r)p_k(b_k - \Delta_{i_r} - \Delta_{i_r, i_n^*}) + (1-p_r)(1-p_k)\smash{\tilde{p_2}}\smash{\tilde{b_2}}.
    \end{align*}
    This is equivalent to
    \[
        \smash{\tilde{p_2}} \Delta_{i_r}(1-p_k) + \Delta_{i_r, i_n^*}(1-p_k) + (1-p_k)\smash{\tilde{p_2}}\smash{\tilde{b_2}} > \smash{\tilde{p_2}}b_r^k(1-p_k) + p_k(b_r^k - b_k).
    \]
    For $p_k = 1$, this is would equivalent to $0 > b_r^k - b_k$, which is a contradiction. Thus it must hold $p_k < 1$ and therefore
    \[
        \Delta_{i_r} + \frac{\Delta_{i_r, i_n^*}}{\smash{\tilde{p_2}}} + \smash{\tilde{b_2}} > b_r^k + \frac{p_k}{\smash{\tilde{p_2}}(1-p_k)} (b_r^k - b_k).
    \]
    As $\frac{p_k}{\smash{\tilde{p_2}}(1-p_k)} \ge 0$ and $b_r^k - b_k > 0$, it follows
    \begin{equation}
        \Delta_{i_r} + \frac{\Delta_{i_r, i_n^*}}{\smash{\tilde{p_2}}} + \smash{\tilde{b_2}} > b_k.
        \label{Eq:ForwardSwapOfBoxRNotProfitable}
    \end{equation}

    Now suppose swapping box $k$ right after box $r$ is decreasing the utility, then it holds
    \begin{align*}
         &~ p_k(b_k - \Delta_{i_r} - \Delta_{i_r, i_n^*}) + (1-p_k)\smash{\tilde{p_2}}\smash{\tilde{b_2}} + (1-p_k)(1-\smash{\tilde{p_2}})p_r(b_r^k - \Delta_{i_r})\\
        >&~ \smash{\tilde{p_2}}\smash{\tilde{b_2}} + (1-\smash{\tilde{p_2}})p_r(b_r^k - \Delta_{i_r}) + (1-\smash{\tilde{p_2}})(1-p_r)p_k(b_k - \Delta_{i_r}).
    \end{align*}
    This is equivalent to
    \[
        p_r b_k(1-\smash{\tilde{p_2}}) + \smash{\tilde{p_2}} b_k > \smash{\tilde{p_2}}\smash{\tilde{b_2}} + \Delta_{i_r, i_n^*} + \smash{\tilde{p_2}} \Delta_{i_r} + p_r b_r (1-\smash{\tilde{p_2}}).
    \]
    With $b_r^k > b_k$ it follows
    \begin{equation}
        b_k > \smash{\tilde{b_2}} + \Delta_{i_r} + \frac {\Delta_{i_r, i_n^*}}{\smash{\tilde{p_2}}}.
        \label{Eq:BackwardSwapOfBoxRNotProfitable}
    \end{equation}
    Equations \eqref{Eq:ForwardSwapOfBoxRNotProfitable} and \eqref{Eq:BackwardSwapOfBoxRNotProfitable} cannot both be true, as they contradict each other. Therefore, we can iteratively remove the current box $r$ from the set $L$ in $\sigma^*$: If $b_r^k \le b_k$, then we can swap box $r$ to position $i_n$, removing it from $L$. If $b_r^k > b_k$, then we can either swap box $r$ right before box $k$, removing it from $L$, or we can swap box $k$ right after box $r$ and have $L = L_1$ immediately. All these swaps cannot deteriorate the utility for \calP. Eventually, $L = L_1$, and Equation \eqref{Eq:ForwardSwapOfBoxKNotProfitable} shows us again that we can swap box $k$ behind $L_1$ without decreasing the utility, proving the inductive step when $i^* < i$.
\end{proof}

\section{Missing Material from Section~\ref{sec:iid}}
\label{app:iid}

\begin{proof}[Proof of Lemma~\ref{lem:iid-p2}]
    The lemma is clear when $a_0 \ge \faircap{}(0)$ (because then $t_i=0$ for all basic boxes in any contract), so assume $a_0 < \faircap{}(0)$ in the following.
 
    Suppose there are rounds $\ell, \ell+1$ with basic boxes and $a_0 + t_\ell < \faircap{}(0) = a_0 + t_{\ell+1}$. We show that by swapping these boxes the principal weakly improves. Formally, we set $t_\ell' = t_{\ell+1}$ and $t_{\ell+1}' = t_\ell$ (and $t_i' = t_i$ for all other $i$),
 
	Fix any (random) vector of prizes of all $n$ rounds. Let $j(i)$ denote the outcome of round $i \in [n]$. Furthermore $j_1 = \max_{i = 1}^{\ell-1} j(i)$ denotes the best outcome for $\mathcal A$ until round $\ell-1$. Moreover, let $E(x, j)$ denote the expected utility for $\mathcal P$ when the process starts at round $x$ with prize $j$ being the best outcome for the agent so far.
 
	If box $\ell$ is not reached and the process stopped earlier, the outcome does not change under the modified contract $T'$. Otherwise, we reach round $\ell$ with a current maximum of $a_{j_1} \le \faircap{}(0)$ and distinguish between multiple cases:
	\begin{enumerate}
		\item $j(\ell) = j(\ell+1) = 0$.
		Then the process stops in round $\ell+1$ with utility $v - t_{\ell+1}$ for $\mathcal P$. 
		Under $T'$, the process stops in round $\ell$ with utility $v - t_{\ell}'$, which is the same.
		\item $j(\ell) = 0$ and $1 \le j(\ell+1) \le j^*$.
		Then round $\ell+2$ is reached. The expected utility of $\mathcal P$ is given by $E(\ell+2, j')$, where $j'$ is the prize from $\{0, j_1, j(\ell+1)\}$ with maximum utility for \calA.
		Under $T'$, the process stops in round $\ell$ with utility $v - t_\ell'$ for $\mathcal P$.
		\item $j(\ell) = 0$ and $j(\ell+1) > j^*$.
		The the process stops in round $\ell+1$ with utility 0 for $\mathcal P$.
		Under $T'$, the process stops in round $\ell$ with utility $v - t_\ell'$ for $\mathcal P$.
		\item $1 \le j(\ell) \le j^*$ and $j(\ell+1) = 0$.
		The the process stops in round $\ell+1$ with utility $v - t_{\ell+1}$ for $\mathcal P$.
		Under $T'$ however, round $\ell + 2$ is reached. The expected utility of $\mathcal P$ is given by $E(\ell+2, j')$, where $j'$ is the prize from $\{0, j_1, j(\ell)\}$ with maximum utility for \calA.
		\item $j(\ell) > j^*$ and $j(\ell+1) = 0$. Then the process stops after round $\ell$ with utility 0 for the principal under both $T$ and $T'$.
		\item $j(\ell) \neq 0$ and $j(\ell+1) \neq 0$. In this case, outcome 0 is not drawn in either round, so the utility is the same under both $T$ and $T'$, as the outcome that is taken is not influenced by the payment.
	\end{enumerate}
	As we consider i.i.d.\ prizes, cases (2) and (4) have identical probabilities. In all other cases, $\mathcal P$ weakly benefits from the modified contract $T'$.
	By applying the swap argument iteratively, the lemma follows.
\end{proof}

\begin{proof}[Proof of Lemma~\ref{lem:iid-p1}]
	Suppose there is some $\ell \in [k-1]$ with $\faircap{\ell} > \faircap{\ell+1}$.
	Then there are some $r, s \in [m]$ with $\faircap{\ell} = a_r > a_s = \faircap{\ell+1}$.
	For any $i \in [n]$, let $R_i$ denote the event that box $i$ is opened.
	For any $x, y \in [n]$ with $x \le y$, let $E(x, y)$ denote the expected utility that \calP gains from the process from boxes $\{x,x+1,\ldots,y-1,y\}$, conditioned on $R_x$. 
	Furthermore, for any fair cap $\faircap{} > \faircap{}(0)$, let $t(\faircap{})$ denote the payment for prize 0 required to obtain that fair cap.
	Then:
	\[
		E(1,n) = E(1,\ell-1) + \mathbb P(R_\ell) \cdot p \cdot [v - t(a_r)] + \mathbb P(R_{\ell+1}) \cdot p \cdot [v - t(a_s)] + \mathbb P(R_{\ell+2}) \cdot E(\ell+2,n).
	\]
	Note that by definition, the process always stops immediately during Phase 1 if prize 0 is drawn. The fair cap can only be increased to $\faircap{i}$ if $\mathcal A$'s valuation for prize 0 above the fair cap is increased.
	Thus, for all $i=1,\ldots,\ell$, we have $a_0 + t(\faircap{i}) > \faircap{i}$ and, by the order of fair caps, $a_0 + t(\faircap{i}) > \faircap{i+1}$.
 
	This means that $E(1,\ell-1)$ is independent of the fair caps $\faircap{\ell}$ and $\faircap{\ell+1}$, as long as they are not above $\faircap{\ell-1}$ and change set of boxes considered with $E(1,\ell-1)$.
	Analogously, $E(\ell+2,n)$ is independent of these fair caps, as long as they are not below $\faircap{\ell+2}$ and change the set of boxes considered.

	Also, for any round $i \in [k]$ of Phase 1, the process has stopped before round $i$ if and only if a prize value more than $\faircap{i}$ has been found before round $i$. Thus, it holds that $\mathbb P(R_i) = (\sum_{j \in S(\faircap{i})} q_j)^{i-1} =: r(\faircap{}_i, i)$, where $S(\faircap{i}) := \{j \in [m] : a_j \le  \faircap{i}\}$, for any fair cap $\faircap{i}$ and any box $i \in [n]$.
	Thus:
	\[
		E(1,n) = E(1,\ell-1) + r(a_r, \ell) \cdot p \cdot [v - t(a_r)] + r(a_s, \ell+1) \cdot p \cdot [v - t(a_s)] + r(\faircap{\ell+1},\ell+2) \cdot E(\ell+2,n).
	\]
	
	We consider two adaptations of the payments. We either replace $\faircap{\ell}$ by $\faircap{\ell+1}$, or vice versa.
	\begin{itemize}
	    \item Adaptation 1: $\faircap{\ell}' := \faircap{\ell+1} = a_s$ and $\faircap{i}' := \faircap{i}$ for all $i \neq \ell$.

     	This does not change the order of fair caps. Thus the expected utility for $\mathcal P$ is given by
	\[
		E'(1,n) = E(1,\ell-1) + r(a_s, \ell) \cdot p \cdot [v - t(a_s)] + r(a_s, \ell+1) \cdot p \cdot [v - t(a_s)] + r(\faircap{\ell+2},\ell+2) \cdot E(\ell+2,n).
	\]
        \item Adaptation 2: $\faircap{\ell+1}'' := \faircap{\ell} = a_r$ and $\faircap{i}'' := \faircap{i}$ for all $i \neq \ell+1$.
        
	This does not change the order of fair caps. Thus the expected utility for $\mathcal P$ is given by
	\[
		E''(1,n) = E(1,\ell-1) + r(a_r, \ell) \cdot p \cdot [v - t(a_r)] + r(a_r, \ell+1) \cdot p \cdot [v - t(a_r)] + r(\faircap{\ell+2},\ell+2) \cdot E(\ell+2,n).
	\]
	\end{itemize}
 
	We now claim that at least one of these adaptations is always weakly beneficial for $\mathcal P$ by showing that $E(1,n) > E'(1,n)$ implies $E''(1,n) > E(1,n)$. Indeed, it holds
	\begin{align*}
		&& E(1,n) &> E'(1,n) \\
		&\Leftrightarrow& r(a_r, \ell) \cdot p \cdot [v - t(a_r)] &> r(a_s, \ell) \cdot p \cdot [v - t(a_s)]\\
		&\Rightarrow& p \cdot [v - t(a_r)] &> p \cdot [v - t(a_s)]\\
		&\Rightarrow& r(a_r, \ell+1) \cdot p \cdot [v - t(a_r)] &> r(a_s, \ell+1) \cdot p \cdot [v - t(a_s)]\\
		&\Leftrightarrow& E''(1,n) &> E(1,n)
	\end{align*}
	We used that $r > s$, which implies $r(a_r, \ell) \ge r(a_s, \ell)$ as well as $r(a_r, \ell+1) \ge r(a_s, \ell+1)$.
	This proves the claim. 
	
    Applying this adjustment iteratively, we see that there is an optimal contract in which all boxes in Phase 1 have the same fair cap.
\end{proof}

\begin{proof}[Proof of Lemma~\ref{lem:iid-final}]
	Consider an optimal contract with multiple outcomes that are targeted with the payment for outcome 0.
	Then there are $j_1, j_2 \in [j^*]$ with $j_1 < j_2$ such that $n_{j_1} > 0$ and $n_{j_2} > 0$.
	Assume w.l.o.g.\ that the outcomes between $j_1$ and $j_2$ are never targeted, i.e., $n_{j_1 +1} = \dots = n_{j_2 - 1} = 0$.
	
	Suppose we could swap a $\delta$-fraction of boxes from target outcome $j_2$ to $j_1$ (or vice versa), for some $\delta \in \mathbb R$.
	Then let $n_{j_1}' := n_{j_1} + \delta$, $n_{j_2}' := n_{j_2} - \delta$, $n_j' = n_j$ for all $j' \in [j^*] \setminus \{j_1, j_2\}$.
	Accordingly, define $N_j' := \sum_{j' = 1}^j n_{j'}'$.
	
	We want to maximize the expected utility for $\mathcal P$ over $\delta$.
	It holds that
	\begin{align*}
		~&\frac {\partial} {\partial \delta} \left( \sum_{j = 1}^{j^*} v_j \cdot  Q_j^{n} \cdot \left( \left(\frac{p} {Q_j} + 1\right)^{N_j'} - \left(\frac{p} {Q_j} + 1\right)^{N_{j-1}'}\right) \right)\\
		=~& \sum_{j = j_1}^{j_2-1} v_j Q_j^n \ln\left(\frac p {Q_j} + 1\right) \cdot \left(\frac p {Q_j} + 1\right)^{N_j'} - \sum_{j = j_1+1}^{j_2} v_j Q_j^n \ln\left(\frac p {Q_j} + 1\right) \cdot \left(\frac p {Q_j} + 1\right)^{N_{j-1}'}\\
		=~& v_{j_1} Q_{j_1}^n \ln\left(\frac p {Q_{j_1}} + 1\right)\cdot \left(\frac p {Q_{j_1}} + 1\right)^{N_{j_1}'} - v_{j_2} Q_{j_2}^n \ln\left(\frac p {Q_{j_2}} + 1\right) \cdot \left(\frac p {Q_{j_2}} + 1\right)^{N_{j_2 - 1}'} \\
		&+\sum_{j = j_1+1}^{j_2-1} v_j Q_j^n \ln\left(\frac p {Q_j} + 1\right) \left( \left(\frac p {Q_j} + 1\right)^{N_j'} - \left(\frac p {Q_j} + 1\right)^{N_{j-1}'}\right)\\
		=~& v_{j_1} Q_{j_1}^n \ln\left(\frac p {Q_{j_1}} + 1\right)\cdot \left(\frac p {Q_{j_1}} + 1\right)^{N_{j_1}'} - v_{j_2} Q_{j_2}^n \ln\left(\frac p {Q_{j_2}} + 1\right) \cdot \left(\frac p {Q_{j_2}} + 1\right)^{N_{j_2 - 1}'}.
	\end{align*}
	We used that for all $j \in \{j_1 +1, \dots, j_2 -1\}$, it holds that $n_j' = 0$ and therefore $N_j' = N_{j-1}'$.
	
 For any extreme point, it follows that
	\[v_{j_1} Q_{j_1}^n \ln\left(\frac p {Q_{j_1}} + 1\right)\cdot \left(\frac p {Q_{j_1}} + 1\right)^{N_{j_1}'} = v_{j_2} Q_{j_2}^n \ln\left(\frac p {Q_{j_2}} + 1\right) \cdot \left(\frac p {Q_{j_2}} + 1\right)^{N_{j_2 - 1}'}.\]
	Considering the second derivative at any extreme point, it holds that
	\begin{align*}
		~&\frac {\partial^2} {\partial \delta^2} \left( \sum_{j = 1}^{j^*} v_j \cdot  Q_j^{n} \cdot \left( \left(\frac{p} {Q_j} + 1\right)^{N_j'} - \left(\frac{p} {Q_j} + 1\right)^{N_{j-1}'}\right) \right)\\
		=~& v_{j_1} Q_{j_1}^n \ln\left(\frac p {Q_{j_1}} + 1\right)^2\cdot \left(\frac p {Q_{j_1}} + 1\right)^{N_{j_1}'} - v_{j_2} Q_{j_2}^n \ln\left(\frac p {Q_{j_2}} + 1\right)^2 \cdot \left(\frac p {Q_{j_2}} + 1\right)^{N_{j_2 - 1}'}\\
		=~& \ln\left(\frac p {Q_{j_1}} + 1\right) \cdot v_{j_2} Q_{j_2}^n \ln\left(\frac p {Q_{j_2}} + 1\right)\cdot \left(\frac p {Q_{j_2}} + 1\right)^{N_{j_2 - 1}'} 
        \\
        ~&- v_{j_2} Q_{j_2}^n \ln\left(\frac p {Q_{j_2}} + 1\right)^2 \cdot \left(\frac p {Q_{j_2}} + 1\right)^{N_{j_2 - 1}'}\\
		=~& \left(\ln\left(\frac p {Q_{j_1}} + 1\right) - \ln\left(\frac p {Q_{j_2}} + 1\right) \right) \cdot v_{j_2} Q_{j_2}^n \ln\left(\frac p {Q_{j_2}} + 1\right)\cdot \left(\frac p {Q_{j_2}} + 1\right)^{N_{j_2 - 1}'}\\
		>~& 0.
	\end{align*}
	   The last inequality holds because $j_1 < j_2$, so $Q_{j_1} < Q_{j_2}$, and $\frac p {Q_{j_2}} > \frac p {Q_{j_2}}$.
      Furthermore, it holds $\ln\left(\frac p {Q_{j_1}} + 1\right) > \ln\left(\frac p {Q_{j_2}} + 1\right) \ge 0$ due to monotonicity of the logarithm and $\frac p {Q_{j_2}} \ge 0$. All other terms of the product are also non-negative.
	
	As the second derivative is positive at every extremum, there can be at most one local minimum, but no local maximum. 
	Hence, we maximize the function by choosing either the smallest or largest possible value for $\delta$. Therefore, it represents an improvement to either decrease $n_{j_1}$ to 0 and increase $n_{j_2}$ by $n_{j_1}$, or vice versa. Applying this insight repeatedly, the statement of the lemma follows.
\end{proof}

\end{document}